\documentclass[aoas,preprint]{imsart}
\pdfoutput=1 
\setattribute{journal}{name}{} 
\usepackage{pdfpages}

\usepackage{amsthm,amsmath,amssymb}
\RequirePackage[numbers]{natbib}
\usepackage{graphicx}
\usepackage{multirow}
\usepackage{lscape}
\usepackage{enumerate}
\usepackage{booktabs}
\RequirePackage[colorlinks,citecolor=blue,urlcolor=blue]{hyperref}

\startlocaldefs
\numberwithin{equation}{section}
\theoremstyle{plain}
\newtheorem{thm}{Theorem}[section]
\endlocaldefs

\newtheorem{lemma}[thm]{Lemma}

\newtheorem{theorem}[thm]{Theorem}

\newtheorem{remark}[thm]{Remark}
\newtheorem{corollary}[thm]{Corollary}

\newcommand{\RSS}{\mathrm{RSS}}
\newcommand{\df}{\mathrm{df}}
\newcommand{\err}{\mathrm{err}}

\newcommand{\E}{\mathbb{E}}

\newcommand{\eqd}{\overset{d}{=}}
\newcommand{\toas}{\overset{\mathit{a.s.}}{\to}}

\newcommand{\Normal}{\mathcal{N}}

\newcommand{\reals}{\mathbb{R}}

\newcommand{\trans}{\mathrm{T}}

\newcommand{\vectorsymbol}{}
\newcommand{\vzero}{\vectorsymbol{0}}

\newcommand{\ve}{\vectorsymbol{e}}

\newcommand{\vs}{\vectorsymbol{s}}
\newcommand{\vt}{\vectorsymbol{t}}

\newcommand{\vv}{\vectorsymbol{v}}

\newcommand{\vectorhatsymbol}[1]{{\hat{#1}}}

\newcommand{\vhv}{\vectorhatsymbol{v}}

\newcommand{\vectortildesymbol}[1]{{\tilde{#1}}}
\newcommand{\vtv}{\vectortildesymbol{v}}

\newcommand{\matrixsymbol}{}
\newcommand{\mzero}{\matrixsymbol{0}}
\newcommand{\mone}{\matrixsymbol{1}}
\newcommand{\mA}{\matrixsymbol{A}}
\newcommand{\mB}{\matrixsymbol{B}}
\newcommand{\mC}{\matrixsymbol{C}}
\newcommand{\mD}{\matrixsymbol{D}}
\newcommand{\mE}{\matrixsymbol{E}}

\newcommand{\mH}{\matrixsymbol{H}}
\newcommand{\mI}{\matrixsymbol{I}}

\newcommand{\mO}{\matrixsymbol{O}}
\newcommand{\mP}{\matrixsymbol{P}}
\newcommand{\mQ}{\matrixsymbol{Q}}
\newcommand{\mR}{\matrixsymbol{R}}
\newcommand{\mS}{\matrixsymbol{S}}

\newcommand{\mU}{\matrixsymbol{U}}
\newcommand{\mV}{\matrixsymbol{V}}

\newcommand{\mX}{\matrixsymbol{X}}
\newcommand{\mY}{\matrixsymbol{Y}}
\newcommand{\mZ}{\matrixsymbol{Z}}

\newcommand{\mGamma}{\matrixsymbol{\Gamma}}

\newcommand{\mSigma}{\matrixsymbol{\Sigma}}

\newcommand{\matrixtildesymbol}[1]{{\tilde{#1}}}

\newcommand{\mtV}{\matrixtildesymbol{V}}

\newcommand{\matrixhatsymbol}[1]{{\hat{#1}}}
\newcommand{\mhA}{\matrixhatsymbol{A}}
\newcommand{\mhB}{\matrixhatsymbol{B}}

\newcommand{\mhD}{\matrixhatsymbol{D}}
\newcommand{\mhE}{\matrixhatsymbol{E}}

\newcommand{\mhU}{\matrixhatsymbol{U}}
\newcommand{\mhV}{\matrixhatsymbol{V}}

\newcommand{\mhY}{\matrixhatsymbol{Y}}

\newcommand{\mhGamma}{\matrixhatsymbol{\Gamma}}

\startlocaldefs
\numberwithin{equation}{section}
\endlocaldefs

\begin{document}

\begin{frontmatter}

\title{Degrees of freedom for combining regression with factor analysis }
\runtitle{Degrees of freedom}

\author{\fnms{Patrick} \snm{O. Perry}\ead[label=e1]{pperry@stern.nyu.edu}\thanksref{t1}}, 
\author{\fnms{Natesh} \snm{Pillai}\ead[label=e2]{pillai@fas.harvard.edu}\thanksref{t2}}
\affiliation{New York University\thanksmark{t1}}
\affiliation{Harvard University\thanksmark{t2}}


\runauthor{Perry and Pillai}

\begin{abstract}

In the AGEMAP genomics study, researchers were interested in detecting genes
related to age in a variety of tissue types.  After not finding many
age-related genes in some of the analyzed tissue types, the study was
criticized for having low power \citep{Land:08}.  It is possible that the low
power is due to the presence of important unmeasured variables, and indeed we
find that a latent factor model appears to explain substantial variability not
captured by measured covariates.  We propose including the estimated latent
factors in a multiple regression model.  The key difficulty in doing so is
assigning appropriate degrees of freedom to the estimated factors to obtain
unbiased error variance estimators and enable valid hypothesis testing.  When
the number of responses is large relative to the sample size, treating the
estimated factors like observed covariates leads to a downward bias in the
variance estimates.  Many \emph{ad-hoc} solutions to this problem have been
proposed in the literature without the backup of a careful theoretical
analysis.  Using recent results from random matrix theory, we derive a simple,
easy to use expression for degrees of freedom. Our estimate gives a principled
alternative to \emph{ad-hoc} approaches in common use. Extensive simulation
results show excellent agreement between the proposed estimator and its
theoretical value.  Applying our methodology to the AGEMAP genomics study, we
found an order of magnitude increase in the number of significant genes.
Although we focus on the AGEMAP study, the methods developed in this paper are
widely applicable to other multivariate models, and thus are of independent
interest.

\end{abstract}



\end{frontmatter}
\section{Introduction}\label{S:introduction}
In the AGEMAP genomics study of $M \approx 18,000$ genes measured in $N = 39$ subjects
(mice), researchers are interested in detecting which genes are related to
age~\citep{Zahn:07}.  For each subject-gene pair $ij$, with $1 \leq i \leq N$
and $1 \leq j \leq M$, they measure $y_{ij}$, the log-activation in subject
$i$ of gene $j$; taken together, these measurements form a response
matrix $\mY = [y_{ij}] \in \reals^{N \times M}$.

The researchers have two covariate matrices available.  The row
covariate matrix, $\mX = [x_{ik}] \in \reals^{N \times p}$, encodes
subject-specific attributes.   This matrix has $p = 3$ columns, for an
intercept, the sex, and the age of the subject:
\begin{align*}
  x_{i1} &= 1, \\
  x_{i2} &= \text{Sex of subject $i$ (Female = +1, Male = -1),} \\
  x_{i3} &= \text{Age of subject $i$ (months).} 
\end{align*}
The column covariate matrix, $\mZ = [z_{jl}] \in \reals^{M \times q}$, encodes
response-specific attributes.  This matrix has $q = 2$ columns, for an
intercept and the tissue type of the response:
\begin{align*}
  z_{j1} &= 1, \\
  z_{j2} &= \text{Tissue of response $j$ (Cerebellum = +1, Cerebrum = -1).}
\end{align*}

To model the associations between the covariates and the response, it is
natural to posit existence of row and column coefficient matrices
$\mA = [\alpha_{il}] \in \reals^{N \times q}$ 
and
$\mB = [\beta_{jk}] \in \reals^{M \times p}$ 
which link the covariates to the response via the relation
\[
  \mY = \mX \mB^\trans + \mA \mZ^\trans + \mE,
\]
where $\mE = [\varepsilon_{ij}]\in \reals^{n \times m}$
is a matrix of mean-zero random errors.  The interpretation of $\beta_{j3}$ is
as follows: ``holding sex and subject-specific effects constant, increasing
age by 1 unit (1 month) is associated with increasing expected log activation of gene j
by $\beta_{j3}$ units.''

%

For the AGEMAP study, we would like to answer the question: `` is gene $j$ associated with age if $\beta_{j3}$ is nonzero''.  However, individual components of
$\beta$ are not identifiable, so this is not a workable definition.  Instead, we will say
that gene $j$ is related to age if the age coefficient for that gene differs
from the average age coefficient for all genes of the same tissue type.  More
precisely, we say that gene $j$ is related to age if
\(
  [\mB^\trans \vs^{(j)}]_3 = 0,
\)
where
\(
  \vs^{(j)} = (\mI - \mH_{\mZ}) \ve_j,
\)
with $\mI$ the identity matrix,
$\ve_j$ the $j$\textsuperscript{th} standard basis vector in $\reals^{M}$
and $\mH_{\mZ} = \mZ (\mZ^\trans \mZ)^{-1} \mZ^\trans$.  Alternative
definitions are possible by using weighted versions of the hat matrix
$\mH_{\mZ}$.

Following \citet{Gabr:78}, we estimate the identifiable components of the
coefficient matrices via least squares.  We choose estimates $\mhA$ and
$\mhB$ to satisfy
\begin{align*}
  (\mI - \mH_{\mX}) \mhA &= (\mI - \mH_{\mX}) \mY \mZ (\mZ^\trans \mZ)^{-1}, \\
  (\mI - \mH_{\mZ}) \mhB &= (\mI - \mH_{\mZ}) \mY^\trans \mX (\mX^\trans \mX)^{-1}.
\end{align*}
That is, we find the identifiable components of $\mhA$ by regressing on the
row residuals from a column regression of $\mY$ on $\mX$; we find the
identifiable components of $\mhB$ by regressing on the column residuals from a
row regression of $\mY$ on $\mZ$.
Letting
\(
  \mhY = \mhA \mZ^\trans + \mX \mhB^\trans,
\)
the unidentifiable components can be chosen arbitrarily such that
\[
  \mhE
    \equiv
    \mY - \mhY
    = (\mI - \mH_{X}) \mY (\mI - \mH_{Z});
\]
one possibility is to take
\(
  \mH_{\mZ} \mhB = \mH_{\mZ} \mY^\trans \mX (\mX^\trans \mX)^{-1}
\)
and
\(
  \mH_{\mX} \mhA = \mzero.
\)
When the estimates are chosen in this manner,  it is easy to show the following. 
%
If $\mX$ is full rank and
\(
  \mY = \mA \mZ^\trans + \mX \mB^\trans + \mE,
\)
where the rows of $\mE$ are independent mean-zero multivariate normal random
vectors with covariance matrix $\mSigma$, then for any $\vs$ is any vector satisfying $\mZ^\trans \vs = \vzero$, the quantities
$\mhB^\trans \vs$ and $\vs^\trans \mhE^\trans \mhE \vs$ are independent with
\begin{align}
  \mhB^\trans \vs
    &\sim
      \Normal\!\big(\mB^\trans \vs, \, \vs^\trans \mSigma \vs \cdot (\mX^\trans \mX)^{-1}\big), \label{eqn:tdist1} \\
  \vs^\trans \mhE^\trans \mhE \vs
    &\sim
      \vs^\trans \mSigma \vs \cdot \chi^2_{N - p}.\label{eqn:tdist2}
\end{align}
%
The main implication of Equations \eqref{eqn:tdist1} and \eqref{eqn:tdist2} is that, if
$\vt$ is any vector and
$\vs^\trans \mB \vt = 0$, then the test statistic
\[
  T(\vs, \vt)
  \equiv
  \frac{
    \sqrt{N - p} \cdot \vs^\trans \mhB \vt
  }{
    \{
      \vs^\trans \mhE \mhE^\trans \vs
      \cdot
      \vt^\trans (\mX^\trans \mX)^{-1} \vt
    \}^{1/2}
  }
\]
is $t$-distributed with $N - p$ degrees of freedom.  This facilitates
hypothesis testing on the components of $\mB^\trans \vs$.

\subsection{The problem}

With all this machinery in place, suppose that we want to test whether a
particular gene, Mm.71015 (Cerebellum) is related to age.  First, we
fit $\mhY = \mhA \mZ^\trans + \mX \mhB^\trans$ via least squares, and we
calculate the residuals $\mhE = \mY - \mhY$.  We set
$\vs = (\mI - \mH_{\mZ}) \ve_j$, where $j$ is the index of
Mm.71015 (Cerebellum).  The estimate (standard error) of the age component
of $\mhB^\trans \vs$ is 0.018 (0.014); the $T$ statistic is 1.36, with 36
degrees of freedom.  Apparently, the gene is not significantly related to
age.

\begin{figure}
\centering
\includegraphics[scale = 0.7]{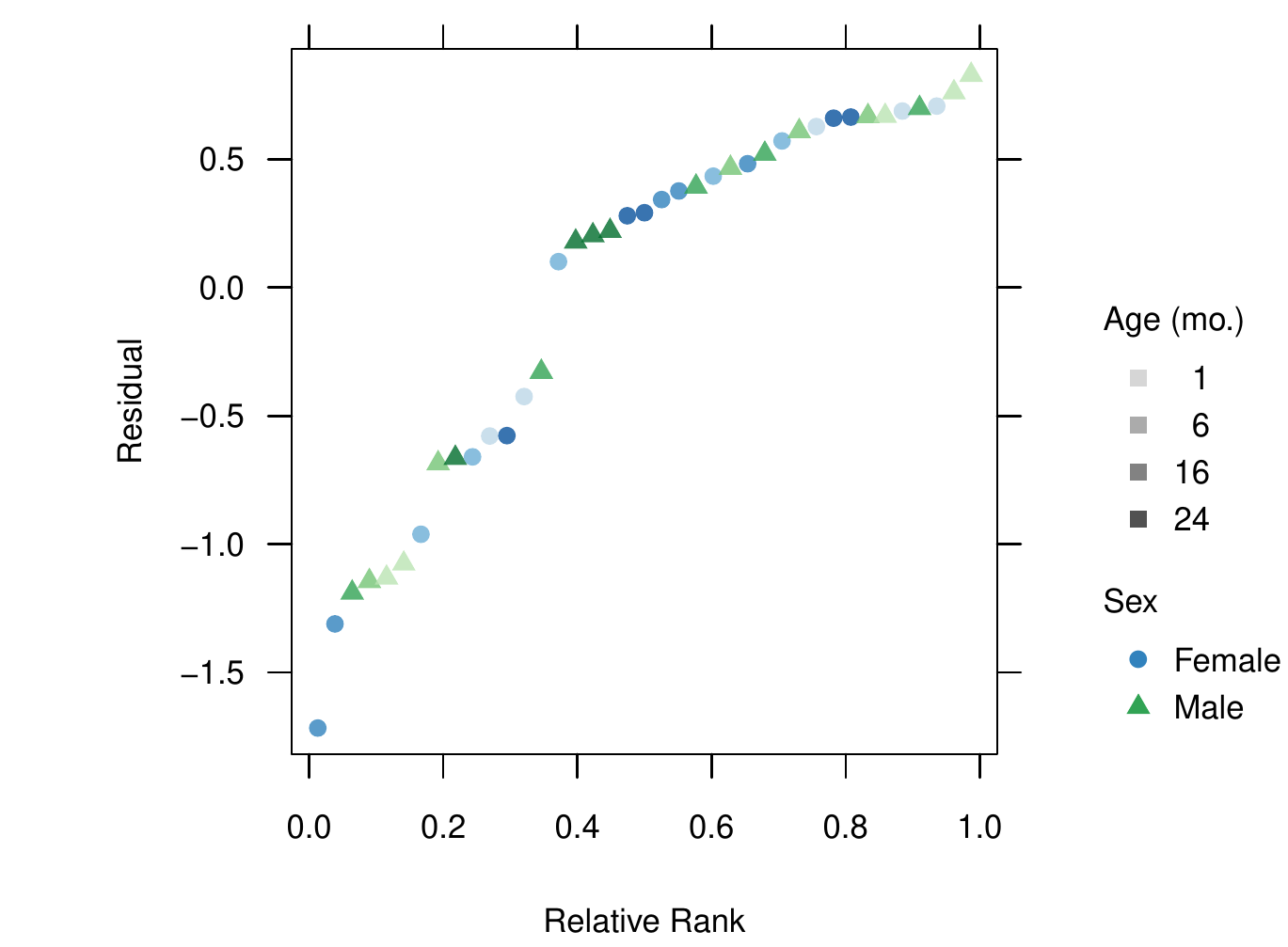}
\caption{\textsc{Residuals Reveal Latent Structure.}
Residuals from the regression of log activation on sex and age in
gene Mm.71015 (Cerebellum).  Two clusters of subjects are apparent.
}\label{fig:agemap-1gene-resid}
\end{figure}

Visually inspecting the elements of the residual component $\mhE \vs$ reveals
a problem with the modeling assumptions (Figure~\ref{fig:agemap-1gene-resid}).
Specifically, our analysis relies on the elements of the regression error
component $\mE \vs$ being independent mean-zero normal random variables.  As
evidenced by the multi-modal structure in the residuals, the distributional
assumptions on the regression errors seem implausible.

An analysis of all $M$ genes further corroborates the evidence of latent
structure in the residual matrix $\mhE$.  If the model were correctly
specified, then there should be no apparent row-specific structure in the
residual matrix.  However, as Figure~\ref{fig:agemap-pc} demonstrates,
there are clear clusters in the first two principal component scores
computed from $\mhE$.  One cluster of subjects exhibits low response values
across many Cerebrum tissue genes, another cluster exhibits low response values
across many Cerebellum tissue genes, and the remaining cluster has medial
responses for most genes, regardless of tissue
type~(Figure~\ref{fig:agemap-scatter-grid}).

\begin{figure}
\centering
\includegraphics[scale = 0.7]{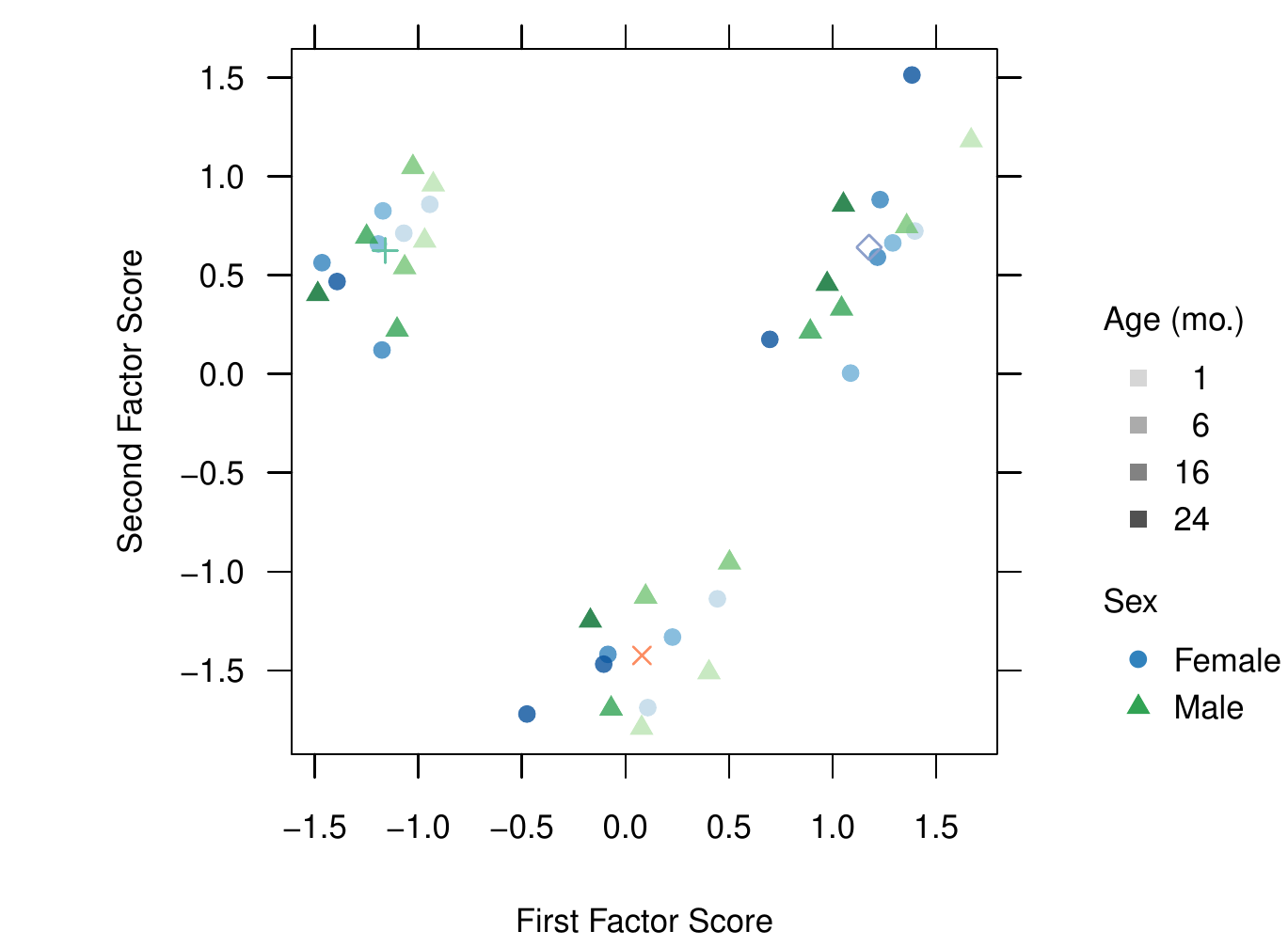}
\caption{\textsc{Residual Matrix Reveals Latent Structure.}
First two principal component scores computed from the residuals after
regressing gene response on age and gender.  Three clusters of individuals
are apparent.
}\label{fig:agemap-pc}
\end{figure}

\begin{figure}
\centering
\includegraphics[scale = 0.75]{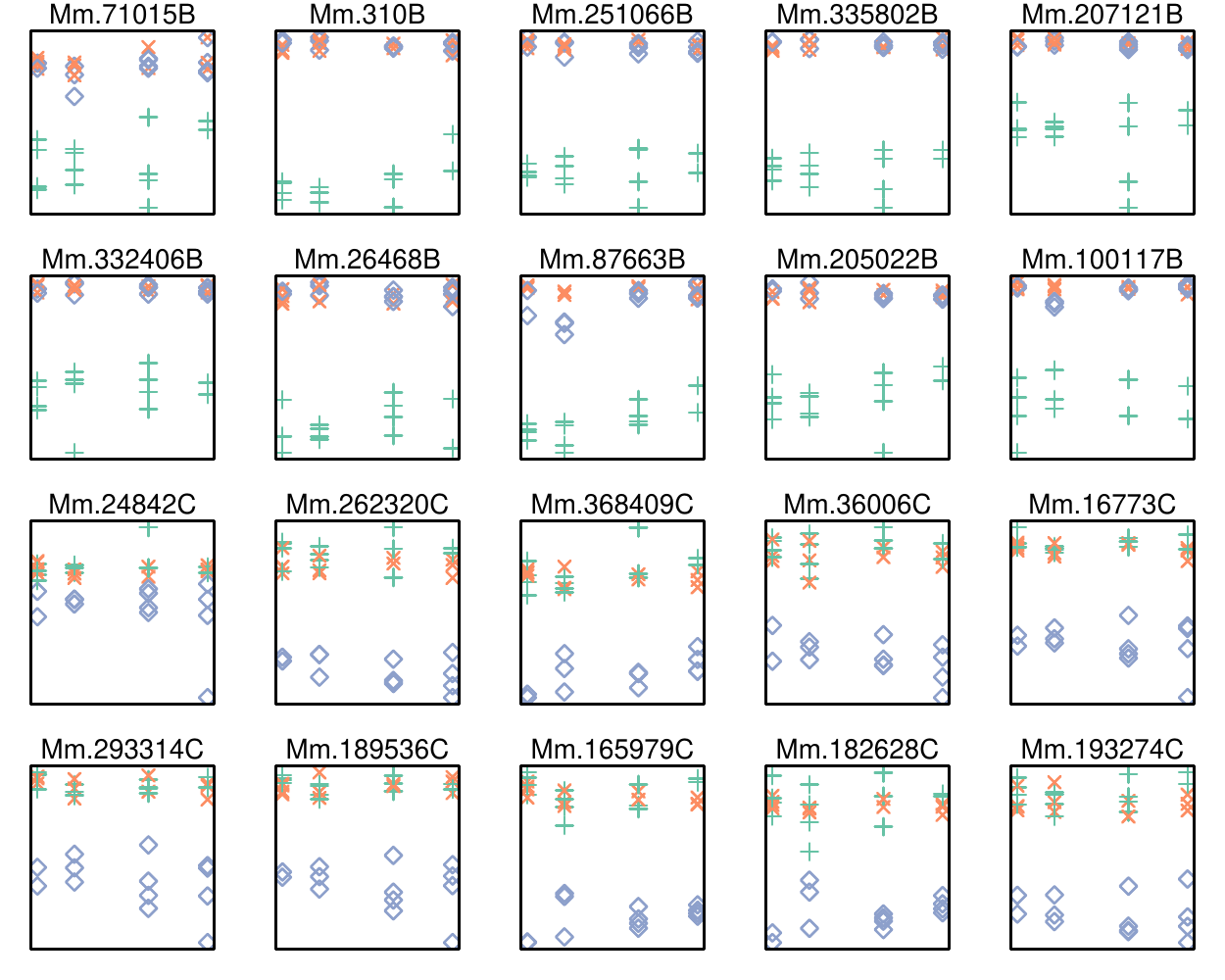}
\caption{\textsc{Mouse Clusters Exhibit Different Response Behaviors.}
Scatterplots of age versus log gene activation for ten Cerebrum genes (B) and
ten Cerebellum genes (C), with subject colors determined from the clusters
identified in Figure~\ref{fig:agemap-pc}.
}\label{fig:agemap-scatter-grid}
\end{figure}

The principal components analysis of the residual matrix hints at the
existence of latent subject-specific covariates.  It is likely that there is
some $N \times r$ matrix $\mU$ of unobserved subject-specific covariates, and
an $M \times r$ matrix $\mV$ of coefficients such that
\[
  \mY = \mA \mZ^\trans + \mX \mB^\trans + \mU \mV^\trans + \mE.
\]

To make the model identifiable, we require that $\mU^\trans \mX = \mzero$ and
$\mV^\trans \mZ = \mzero$.  Without the identifiability assumption, the least
squares estimates of $(\mI - \mH_{\mX}) \mA$ and $(\mI - \mH_{\mZ}) \mB$ will
be biased by $(\mI - \mH_{\mX}) \mU \mV^\trans \mZ (\mZ^\trans \mZ)^{-1}$ and
$(\mI - \mH_{\mZ}) \mV \mU^\trans \mX (\mX^\trans \mX)^{-1}$.  In fact, since
we never perform inference on $\mA$, the identifiability assumption on
$\mV$ is inconsequential.  The constraint $\mU^\trans \mX = \mzero$, amounts
to a requirement that the latent subject-specific covariates be uncorrelated
with the columns of $\mX$.  Even though the identifiability assumption seems
strong, making this assumption is \emph{less restrictive} than assuming that
$r = 0$ (that is, assuming that there are no latent factors which are
correlated with the response).

With the estimates $\mA$ and $\mB$ the same as in the case with
no latent factors, the least squares estimates
of $\mU$ and $\mV$ can be obtained from the leading $r$ terms of the singular
value decomposition of the residual matrix $\mhE$.  With estimated
latent factors having scores $\mhU$ and loadings
$\mhV$, this gives an adjusted residual matrix
$\mhE_{1} = \mhE - \mhU \mhV^\trans$.  Forming the adjusted residual matrix in
this way is equivalent to treating $\mhU$ like observed row covariates.

In the more general latent factor model, to test whether gene $j$ is related
to age after adjusting for observed and latent mouse-specific covariates, we
can base a test on the estimate $\mhB^\trans \vs$, which, even in the presence
of latent factors, is distributed according to~\eqref{eqn:tdist1}.
Unfortunately, such a test requires an estimate of the variance component
$\vs^\trans \mSigma \vs$, which is not readily available; with latent
factors,~\eqref{eqn:tdist2} no longer holds.

To estimate the variance component from the residual matrix $\mhE_1$, we need
to know how many ``degrees of freedom'' are associated with estimating and
adjusting for the latent factor term $\mU \mV^\trans$.  Intuitively, if $\mU$
were known, so that $\mhU$ was equal to $\mU$, then $\E[\vs^\trans \mhE_1^\trans \mhE_1 \vs]$
would be equal to $(N - p - r) \vs^\trans \mSigma \vs$; adjusting for the
factor $\mU \mV^\trans$ would take $r$ degrees of freedom.  In the general
case, when $\mU$ is estimated from data, we define the degrees of
freedom to be the quantity $\df(\vs)$ satisfying the equation
\begin{equation}
  \E[\vs^\trans \mhE_1^\trans \mhE_1 \vs]
  =
  \{ N - p - \mathrm{df}(\vs) \}
  \cdot
  \vs^\trans \mSigma \vs.
\end{equation}
If we knew $\df(\vs)$, then we could get an unbiased estimate of $\vs^\trans
\mSigma \vs$, specifically the quantity
\begin{equation}
  \hat \sigma^2(\vs)
  =
  \frac{
    \vs^\trans \mhE_1^\trans \mhE_1 \vs
  }{
    N - p - \mathrm{df}(\vs)
  }.
\end{equation}
This would facilitate a test on the age component of $\mB^\trans \vs$.



\subsection{Previous work}\label{sec:pw}
Multivariate response data, like the AGEMAP study, is prevalent in diverse applications ranging from
agriculture to econometrics to psychology \citep{BG96,EK95, SL01, SW99, Zahn:07}.  In these
applications, the response variate can be conceived of as a matrix,
$\mY = [y_{ij}] \in \reals^{n \times m}$; the goal is to explain the
variability in the response, and to uncover the relationship between $\mY$ and
observed covariates.

Given row and column covariate matrices
$\mX = [x_{ik}] \in \reals^{n \times p}$ and
$\mZ = [z_{jk}] \in \reals^{m \times q}$,
one natural model linking the covariates to the response is
\begin{equation}\label{E:regression-model}
  \mY = \mX \mB^\trans + \mA \mZ^\trans + \mE,
\end{equation}
where
$\mB = [\beta_{jk}] \in \reals^{m \times p}$
and
$\mA = [\alpha_{ik}] \in \reals^{n \times q}$
are
unknown coefficient matrices and
$\mE = [ \varepsilon_{ij} ] \in \reals^{n \times m}$
is a matrix of random errors.  The full coefficient matrices are
not identifiable, as can be seen by the identity
\(
  \mX \mB^\trans + \mA \mZ^\trans
  =
  \mX (\mB + \mZ \mC)^\trans + (\mA - \mX \mC^\trans) \mZ^\trans.
\)
However, for any
vector $\vs$ orthogonal to the column covariates ($\mZ^\trans \vs = 0$) it is
possible to identify $\mB^\trans \vs$;  similarly, 
for any vector $\vt$ orthogonal to the row covariates
it is possible to identify
$\mA^\trans \vt$.

As explained before, the model~\eqref{E:regression-model} is often inadequate
for explaining observed data.  It is implausible that all sources of
variability have been observed.  To this end, one popular approach is to posit
existence of $r$ latent factors, such that
\begin{equation}\label{E:bilinear-model}
  \mY = \mX \mB^\trans + \mA \mZ^\trans + \mU \mV^\trans + \mE,
\end{equation}
where $\mU = [u_{ik}] \in \reals^{n \times r}$ is thought of as a matrix of
row scores and $\mV = [v_{jk}] \in \reals^{m \times r}$ is a matrix
of column loadings.  Model~\eqref{E:bilinear-model}, which combines regression
and factor analysis, is known as a bilinear model \citep{Gabr:78}.

The bilinear model has appeared in various forms, and it has a long history
dating back to~\citet{FishMack:23}, with notable early contributions by
\citet{Coch:43} and \citet{Will:52}.  The model was relatively obscure until
\citeauthor{Tuke:62}, unaware of his predecessors, suggested combining
regression and factor analysis in his essay, ``The future of data analysis''
\citeyearpar{Tuke:62}.  This inspired \citet{Goll:68} and
\citet{Mand:69,Mand:71} to independently reinvent Williams' version of the
latent factor model.  At this point, the models and its variants saw broader
adoption.  \citet{Free:73} surveys the early history, and
\citet{bartholomew2011latent} give a more recent history.

The special case when $\mX = \mone_{n,1}$ and $\mZ = \mone_{m,1}$ continues to
be popular in agronomy, where it is known as the additive main effects with
multiplicative interaction (AMMI) model \citep{Cros:etal:02,Dias:Krza:03}.
Other recent work on related latent factor models include papers by
\citet{vanE:95}, \citet{Corn:Seye:97}, \citet{Gabr:98}, \citet{West:03},
\citet{Hoff:07}, \citet{Carv:etal:08}, \citet{Leek:Stor:08},
\citet{Frig:etal:09}, and \citet{Sun:etal:12}.

Usually, the parameters of a bilinear model are estimated via least squares.
After this estimation, to perform inference on the coefficients, we need an
estimate of the error variance.  To this end, a persistent challenge
is the assignment of the appropriate ``degrees of
freedom'' to estimates of the factor term.  The statistical literature
remains divided on this issue:
\begin{itemize}
  \item \citet{Goll:68} proposed a parameter-counting scheme.  The least squares
    estimate of the first column of $\mU$, which is orthogonal to $\mX$, has
    $n$ components but satisfies $p$ constraints.  Similarly, the least squares
    estimate of the first column of $\mV$ has $m$ components but satisfies $q$
    constraints.  The scale of either estimated column can be fixed without
    affecting the overall fit. Thus, Gollub allocates $(n - p) + (m - q) - 1$
    degrees of freedom to the first term of the estimated factor.
    Similarly,
    he allocates $(n - p - k) + (m - q - k) - 1$ degrees of freedom to the
    $(k-1)$\textsuperscript{th} estimated factor term.

  \item \citet{Mand:71} noted that when there are no true factors
    ($r = 0$), if the elements of $\mE$ are independent
    normal random variables with common variance, then the
    squared Frobenius norm (sum of squares) of the $k$\textsuperscript{th} estimated latent
    factor is distributed as $\lambda_k$, the $k$\textsuperscript{th} largest eigenvalue of an
    $(m - q) \times (m - q)$ white Wishart matrix with $n - p$ degrees of
    freedom.  Thus, Mandel proposes allocating $\E[\lambda_k]$ degrees of
    freedom to the $k$\textsuperscript{th} estimated factor term, which he computes via
    Monte Carlo simulation.

  \item More recent approaches do not assume that the elements of $\mE$ have a
    common variance, and they use iterative schemes to estimate the factors
    and the noise variances simultaneously.  Essentially, these approaches
    treat the estimated factor scores $\mhU$ like observed covariates $\mX$.
    They either treat the factor loadings as fixed effects, allocating $m$
    degrees of freedom to the $k$\textsuperscript{th} estimated factor
    \citep{Leek:Stor:08,Sun:etal:12}, or they treat the factor loadings as
    random effects which may result in a smaller estimate for the degrees of freedom   
     \citep{Frig:etal:09}.
\end{itemize}
In agronomy and psychometrics applications, with smaller sample sizes,
Gollob's estimate is the more popular method \citep{Dias:Krza:03}; Mandel's
assumption of no true factors is seen as inappropriate.  In genomics
applications, the issues of adjusting for degrees of freedom do not receive much attention,
likely due to an implicit assumption that with large sample sizes, the
adjustment is unimportant.

\subsection{Our contribution}
We first show that that the general latent factor degree of freedom problem
can be reduced to the covariate-free case ($p = q = 0$).

Next, we bring recent developments in random matrix theory to bear on the degrees of freedom problem. In particular, we derive an expression for the expected value of the residual sum of squares. Using this, we then derive conservative estimates for degrees of freedom that are valid when the problem dimensions are large.  Even though
these estimates rely on asymptotic approximations, we observe them to be
accurate for sizes as small as $n = 10$ and $m = 50$.

In the context of linear regression, after dividing by the correct degrees of freedom, the test statistic has the usual t-distribution.  We do not have such a theoretical result in our context. However, we present simulation results showing that our test statistic, after properly adjusting for the degrees of freedom, also has  a corresponding t-distribution. This issue needs further theoretical investigation.\par
Finally, when the data sets are large, our method agrees with most of the other ad-hoc approaches presented in 
Section \ref{sec:pw}, thus bringing theoretical justification to these methods as well. Our results from simulations and the real data example from AGEMAP study also show that not adjusting for the extra degrees of freedom may result in significant loss of power.  In fact, the original analysis of the AGEMAP dataset conducted by \citet{Zahn:07} was criticized
for its low power \cite{Land:08}. Although our analysis in this paper was motivated by the AGEMAP study \citep{Zahn:07}, our methodology for testing from this paper extends easily to other problems. 

The rest of the paper is organized as follows. In Section~\ref{S:reduction}
we reduce the estimation problem to one in which there are no covariates.  In
Section~\ref{S:df}, we derive analytically an asymptotic expression for the
degrees of freedom, which we verify in Section~\ref{S:simulation}.  Next, in
Section~\ref{S:df-estimate}, we propose a conservative degrees of freedom
estimator.  We discuss the implication of our estimator in
the AGEMAP problem Section~\ref{S:AGEMAP}, and we close with a short discussion in
Section~\ref{S:summary}.


\section{Reduction to Covariate-Free Case}\label{S:reduction}
In this section, we show that it is sufficient to consider the case when $p = q = 0$.  
Consider the model
\[
  \mY = \mA \mZ^\trans + \mX \mB^\trans + \mU \mV^\trans + \mE,
\]
where $\mY \in \reals^{N \times M}$, $\mX \in \reals^{N \times p}$,
$\mZ \in \reals^{M \times q}$, and $\mU \mV^\trans$ has rank $r$.  Suppose
that identifiability constraints $\mX^\trans \mU = \mzero$ and
$\mZ^\trans \mV = \mzero$ hold, and
that $\mX$ and $\mZ$ have full column ranks. 
Assume that the rows of $\mE$ are independent
mean-zero multivariate normal random vectors with covariance $\mSigma$.  
Let $\vs$ be a test direction satisfying $\mZ^\trans \vs = \vzero$, and 
define $\sigma^2(\vs) = \vs^\trans \mSigma \vs$.

Take $\mhA$, $\mhB$, $\mhU$, and $\mhV$ to be the least squares estimates
of the parameters with $\hat r$ estimated latent factors and let
$\mhE = \mY - (\mhA \mZ^\trans + \mX \mhB^\trans + \mhU \mhV^\trans)$ be
the residual matrix.  Define residual degrees of freedom
\[
  \df_\mathrm{resid}(\vs) = \E( \vs^\trans \mhE^\trans \mhE \vs) / \sigma^2(\vs).
\]

Let $\mX = \mQ_1 \mR$ be the polar decomposition of $\mX$; that is,
$\mQ_1 \in \reals^{N \times p}$ is a matrix with orthonormal columns, and
$\mR \in \reals^{p \times p}$ is symmetric and positive definite.  Similarly,
let $\mZ = \mP_1 \mS$ be the polar decomposition of $\mZ$.  Choose $\mQ_2$ and
$\mP_2$ such that
\(
  \mQ
  =
  [
  \begin{matrix}
    \mQ_1 & \mQ_2
  \end{matrix}
  ]
\)
and
\(
  \mP
  =
  [
  \begin{matrix}
    \mP_1 & \mP_2
  \end{matrix}
  ]
\)
are orthogonal matrices.
Set $\mY_{22} = \mQ_2^\trans \mY \mP_2$,
$\mU_2 = \mQ_2^\trans \mU$, $\mV_2 = \mP_2^\trans \mV$, and
$\mE_{22} = \mQ_2^\trans \mE \mP_2$,
so that the reduced model holds:
\[
  \mY_{22} = \mU_2 \mV_2^\trans + \mE_{22},
\]
where
$\mY_{22} \in \reals^{n \times m}$ and $\mU_2 \mV_2^\trans$ has rank $r$, with
$n = N - p$ and $m = M - q$.
Note that the rows of $\mE_{22}$ are
independent mean-zero multivariate normal random vectors with covariance
$\mSigma_{22} = \mP_2^\trans \mSigma \mP_2$.
Define $\vs_2 = \mP_2^\trans \vs$ to be
the test direction for the reduced model,
which satisfies the relation
$\vs_2^\trans \mSigma_{22} \vs_2 = \sigma^2(\vs)$.

Take $\mhU_2$ and $\mhV_2$ to be the least squares estimates from
the reduced model for $\mY_{22}$, with $\hat r$ estimated latent factors, and let
$\mhE_{22} = \mE_{22} - \mhU_2 \mhV_2^\trans$ be the residual matrix.  Define
reduced model residual degrees of freedom
\[
  \df_\mathrm{resid}^{(2)}(\vs_2) = \E(\vs_2^\trans \mhE_{22}^\trans \mhE_{22} \vs_2) / \sigma^2(\vs).
\]
\begin{theorem}\label{thm:nocov}
Under the above conditions, we have
\(
  \df_\mathrm{resid}(\vs) = \df_\mathrm{resid}^{(2)}(\vs_2).
\)
\end{theorem}
\begin{proof}
Without loss of generality, redefine $\mA$ and $\mB$ to reparametrize
the model as
\[
  \mY = \mA \mZ^\trans + \mX \mB^\trans + \mX \mGamma \mZ^\trans + \mU \mV^\trans + \mE,
\]
where $\mX^\trans \mA = \mzero$ and $\mZ^\trans \mB = \mzero$.

We perform a change of bases and put the model in block form:
\begin{equation}\label{E:response-block}
  \begin{bmatrix}
    \mQ_1^\trans \\
    \mQ_2^\trans
  \end{bmatrix}
  \mY
  \begin{bmatrix}
    \mP_1 & \mP_2 \\
  \end{bmatrix}
  =
  \begin{bmatrix}
    \mY_{11} & \mY_{12} \\
    \mY_{21} & \mY_{22}
  \end{bmatrix}
  =
  \begin{bmatrix}
    \mR \mGamma \mS^\trans & \mR \mB_2^\trans \\
    \mA_2 \mS^\trans       & \mU_2 \mV_2^\trans
  \end{bmatrix}
  +
  \begin{bmatrix}
    \mE_{11} & \mE_{12} \\
    \mE_{21} & \mE_{22}
  \end{bmatrix},
\end{equation}
where $\mA_2 = \mQ_2^\trans \mA$ and
$\mB_2 = \mP_2^\trans \mB$ are the identifiable components of the regression
coefficients, $\mU_2 = \mQ_2^\trans \mU$ and
$\mV_2 = \mP_2^\trans \mV$ are the identifiable factor components,
and $\mY_{kl} = \mQ_k^\trans \mY \mP_l$ and $\mE_{kl} = \mQ_k^\trans \mE \mP_l$
for $k = 1, 2$ and $l = 1,2$.


From Equation~\eqref{E:response-block} it is apparent that the least squares
estimates of the coefficients are
$\mhA = \mQ_2 \mS^{-1} \mY_{21}$,
$\mhB = \mP_2 \mY_{12}^\trans \mR^{-\trans}$,
and
$\mhGamma = \mR^{-1} \mY_{11} \mS^{-\trans}$.  The regression residuals
$\mhE_0 = \mY - \mhA \mZ^\trans - \mX \mhB^\trans - \mX \mhGamma \mZ^\trans$
satisfy
\begin{equation*}
  \begin{bmatrix}
    \mQ_1^\trans \\
    \mQ_2^\trans
  \end{bmatrix}
  \mhE_0
  \begin{bmatrix}
    \mP_1 & \mP_2 \\
  \end{bmatrix}
  =
  \begin{bmatrix}
    \mzero & \mzero \\
    \mzero & \mU_2 \mV_2^\trans + \mE_{22}
  \end{bmatrix}.
\end{equation*}
Thus, the least squares estimate
$\mhU \mhV^\trans$ obtained from the rank~$\hat r$ singular
value decomposition of $\mhE_0$ is equal to
$\mQ_2 \mhU_2 \mhV_2^\trans \mP_2^\trans$, where $\mhU_2 \mhV_2^\trans$ is
the rank $\hat r$ singular value decomposition
of $\mY_{22} = \mU_2 \mV_2^\trans + \mE_{22}$.

The final residual matrix $\mhE = \mhE_0 - \mhU \mhV^\trans$ is given as
\(
  \mhE
  = \mQ_2 \mhE_{22} \mP_2^\trans,
\)
where
\[
  \mhE_{22}
  = \mY_{22} - \mhU_2 \mhV_2^\trans.
\]
Hence,
\[
  \vs^\trans \mhE^\trans \mhE \vs
  = \vs_2^\trans \mhE_{22}^\trans \mhE_{22} \vs_2
\]
and
\(
  \df_\mathrm{resid}(\vs) = \df_\mathrm{resid}^{(2)}(\vs_2);
\)
the proof is finished.
\end{proof}

\section{Degrees of Freedom}\label{S:df}
In light of Theorem \ref{thm:nocov}, without loss of generality we
will assume that there are no row or column covariates
($p = q = 0$). Our data generating model has $r \geq 0$ true latent factors:
\begin{equation}\label{eqn:regmod}
  \mY = \sqrt{n} \mU \mD \mV^\trans + \mE,
\end{equation}
with $\mU \in \reals^{n \times r}$, $\mV \in \reals^{m \times r}$
having orthonormal columns and a diagonal matrix $\mD \in \reals^{r \times r}$ 
with $[\mD_{kk}] = \sqrt{\mu_k}$ for $k = 1, \dotsc, r$.
We assume that the row vectors of the matrix $\mE$ are  mean-zero multivariate
normal with covariance matrix $\mSigma$.  

The estimates $\mhU$ and $\mhV$ can be obtained from the leading $\hat r$ terms of
the singular value decomposition (SVD) of $\mY$.  We choose the scaling such that
$\sqrt{n} \mhU \mhD \mhV^\trans$ comprises the leading $\hat r$ terms of the SVD of
$\mY$, where $\mhU \in \reals^{n \times \hat r}$ and $\mhV \in \reals^{m \times \hat r}$
have orthonormal columns, and $\mhD \in \reals^{\hat r \times \hat r}$ is diagonal with
$[\mhD]_{kk} = \sqrt{\hat \mu_k}$ for $k = 1, \dotsc, \hat r$.
After adjusting for the estimated latent factors, the residual matrix is
$\mhE = \mY - \sqrt{n} \mhU \mhD \mhV^\trans$.  
The residual sum of squares along the test direction $\vs \in \reals^m$ is given by
\begin{align}\label{eqn:RSSnoise}
  \RSS(\vs)
    &\equiv \vs^\trans \mhE^\trans \mhE \vs.
\end{align}
For $\vs \in \reals^m$, define the degrees of freedom
\begin{align}\label{eqn:dfs}
\df(\vs)  = \E\Big(n -\frac{\RSS(\vs)}{\vs^\trans \mSigma \vs}\Big)
\end{align}
so that 
\[
\E\, \Big(\frac{\RSS(\vs)}{\vs^\trans \mSigma \vs}\Big) = n - \df(\vs).
\]
\begin{lemma}\label{lem:RSS}
For the model in \eqref{eqn:regmod}, 
the residual sum of squares along a test direction $\vs \in \reals^m$ is given by
\begin{align}
\RSS(\vs) = \vs \mE^\trans \mE \vs
  +
  2 \sqrt{n}
  \cdot \vs^\trans \mV \mD \mU^\trans \mE \vs
  +
  n \big(
  \sum_{k=1}^{r}
      \mu_k \cdot (\vv_k^\trans \vs)^2
      -
     \sum_{k=1}^{\hat r} \hat \mu_k \cdot (\vhv_k^\trans \vs)^2
    \big).
\end{align}
\end{lemma}
\begin{proof}
The proof is a straightforward computation and is deferred to the Appendix.
\end{proof}

\subsection{The Noise Case}
In this subsection we assume that there are no true latent factors, \emph{i.e.,} $r=0$, and so
$\mY = \mE$. The rows of $\mE$ are independently distributed according to
$\mathrm{N}(0,\mSigma)$.
\begin{theorem}\label{T:noise-case}
Suppose $\mSigma =  \sigma^2\mI$ for some $\sigma > 0$. If $\lim_{n \rightarrow \infty} \frac{n}{m} = c \in (0, \infty)$, then
\begin{equation}
\df(\vs)  = \hat r\Big( 1 + \sqrt{\frac{n}{m}}\Big)^2 + o(1).
\end{equation}
\end{theorem}
\begin{proof}
For ease of exposition, we first assume $\sigma = 1$.
Applying Lemma \ref{lem:RSS} with $r = 0$ yields
\begin{align*}
  \RSS(\vs) = \vs^\trans \mE^\trans \mE \vs - n \sum_{k=1}^{\hat{r}} \hat \mu_k \cdot (\vhv_k^\trans \vs)^2,
\end{align*}
and thus
\begin{align*}
  \E \,\RSS(\vs) = n\vs^\trans \vs - n \sum_{k=1}^{r} \E \,(\hat \mu_k \cdot (\vhv_k^\trans \vs)^2).
\end{align*}
The matrix $\mE^\trans \mE$ is an $m$-dimensional Wishart matrix with $n$
degrees of freedom and scale parameter $\mSigma$.
The values $\hat \mu_1, \dotsc, \hat \mu_{\hat r}$ are the $\hat r$ largest eigenvalues of
$(1/n) \mE^\trans \mE$.  \citet{Yin:88} show under very general conditions
that, as
$\lim_{n \rightarrow \infty} \frac{n}{m} = c \in (0, \infty)$, 
\[
  \hat \mu_k - (1 + \sqrt{m/n})^2 \toas 0.
\]
Using the results in \cite{pillai2011}, it can be shown that $\E(\hat \mu_k - (1 + \sqrt{m/n})^2)^2$ converges to $0$.
The distribution of $\mhV$ is invariant under multiplication by any $m \times m$ orthogonal
matrix, hence $\E[\vhv^\trans \vs]^2 = (\vs^\trans \vs)/m$.  Also, by a direct calculation using the properties of Haar measure, it can be shown that $(\E[\vhv^\trans \vs]^4)^{1/2} = O((\vs^\trans \vs)/m)$. 
By the above estimates and the Cauchy-Schwartz inequality, it follows that
\[
\E \,(\hat \mu_k \cdot (\vhv_k^\trans \vs)^2)= (1 + \sqrt{m/n})^2(\vs^\trans \vs)/m + o(1)(\vs^\trans \vs)/m
\]
and thus
\[
n \E \,(\hat \mu_k \cdot (\vhv_k^\trans \vs)^2)= (1 + \sqrt{n/m})^2(\vs^\trans \vs) + o(1)(\vs^\trans \vs).
\]
Therefore, if we set
\[
  \df_k(\vs) = (1 + \sqrt{n/m})^2, \qquad 1 \leq k \leq \hat r,
\]
and
\(
  \df(\vs) = \sum_{k=1}^{\hat r} \df_k(\vs),
\)
it follows that
\[
  \{n - \df(\vs)\} - \E[\frac{\RSS(\vs)}{\vs^\trans \vs}]  \rightarrow 0
\]
proving the claim for $\sigma = 1$. The proof for an arbitrary $\sigma >0$ follows
by an identical argument with minor changes.
\end{proof}

\subsection{The Signal Case}
Here we assume that the data are generated according to model \eqref{eqn:regmod} with
$r >0$ latent factors. Without loss of generality, the matrix
$\mD \in \reals^{r \times r}$ is
diagonal with $[\mD_{kk}] = \sqrt{\mu_k}$ for $k = 1, \dotsc, r$ and
$\mu_1 > \cdots > \mu_r > 0$.  Let $\sqrt{n} \mhU
\mhD \mhV^\trans$ be the $\hat r$-term estimated latent factors obtained from the
leading terms of the singular value decomposition of $\mY$, with $\hat r$ not
necessarily equal to $r$.  Let $\vv_k$ denote the $k$\textsuperscript{th} column of $\mV$. 

For any test vector $\vs \in \reals^m$, write
$ \vs = \vs_{\mV} + \vs_{\mV^{\perp}}$
where $\vs_{\mV}, \vs_{\mV^{\perp}}$ respectively denote the projections of $s$ to the column
spaces spanned by $\mV, \mV^{\perp}$. Also recall the degrees of freedom
$\df(\vs)$ given by \eqref{eqn:dfs}.
We need the following lemma, whose proof is given in the Appendix.
\begin{lemma}\label{lem:decomp}
The estimate of the $k$\textsuperscript{th} factor can be decomposed as
\[
  \vhv_k = \sum_{l=1}^{r} \hat \rho_{kl} \, \vv_l + (1 - \hat \rho_k^2)^{1/2} \, \vtv_k,
\]
where
\(
  \hat \rho_k^2 = \sum_{l=1}^{r} \hat \rho_{kl}^2
\),
$\mtV = [ \vtv_1 \dotsc \vtv_{\hat r} ]$ is uniformly distributed on the orthogonal
complement  of $\mV$ and $\hat \rho_{kl}$ is the estimated correlated coefficient.
\end{lemma}
For $k  \leq r$, define the quantity
\begin{equation}\label{E:dfk}
  \df_k(\vs)
    = n \Big(1 - \frac{m}{n \mu_k} - \frac{m}{n \mu_k^2} \Big) \frac{(\vv_k^\trans \vs)^2}{\vs^\trans \vs}
    + \Big( 1 + \frac{1}{\mu_k}\Big)^2 \Big(1 - \frac{(\vv_k^\trans \vs)^2}{\vs^\trans \vs} \Big).
\end{equation}

Before stating our next result, we need the following assumption.\\
\textbf{Assumption A1.} Let $\hat \mu_k$ be the estimated singular values. For $1 \leq k \leq r$,
\begin{align} \label{ass:A1}
\mathbb{E}|\hat \mu_k - \bar \mu_k| = o(n^{-1/2})
\end{align}
where $\bar \mu_k = (\mu_k + 1)\Big(\frac{m}{n\mu_k} + 1\Big)$. Similarly,
\begin{align} \label{ass:A12}
\mathbb{E}|\hat \rho_{kl} - \bar \rho_k| = o(n^{-1/2})
\end{align}
where $\rho_{kk}^2 = \frac{1 - \tfrac{m}{n \mu_k^2}}{1 + \tfrac{m}{n \mu_k}}$,
and $\bar \rho_{kl} = 0$ when $k \neq l$.
\begin{theorem} \label{Thm:signal}
Let Assumption A1 hold.
Suppose $\mSigma = \mI$, $\hat{r} = r$. If
$n/m = c + o(n^{-1/2})$ for some $c \in (0, \infty)$
and if $\mu_r > c^{-1/2}$, then
\begin{equation}\label{eqn:dvs}
\df(\vs)  = \sum_{k=1}^r \df_k(\vs)
          + o\Big(n^{1/2} \sum_{k=1}^r \frac{(\vv_k^\trans \vs)^2}{\vs^\trans \vs} + n^{-1/2} \Big).
\end{equation}
\end{theorem}
\begin{remark}
 Theorem \ref{Thm:signal} holds for $\Sigma = \sigma^2 I$ for any $\sigma > 0$.
For $\sigma \neq 1$, we just need to replace $\mu_k$ with $\mu_k / \sigma^2$ in
Equation~\eqref{E:dfk}.
\end{remark}
\begin{proof}
Since $\hat r = r$, from Lemma \ref{lem:RSS} we obtain
\begin{align}\label{eqn:Erss}
  \E\{\RSS(\vs)\}
  =
  n \cdot \vs^\trans \vs
  +
  n \cdot
  \sum_{k=1}^{r}
    [
      \mu_k \cdot (\vv_k^\trans \vs)^2
      -
      \E\{\hat \mu_k \cdot (\vhv_k^\trans \vs)^2\}
    ].
\end{align}
Thus
\begin{align} \label{eqn:dfnsdef}
\df(\vs) = n -  \frac{\E\{\RSS(\vs)\}}{\vs^\trans \vs} =
- (n/\vs^\trans \vs) \cdot
  \sum_{k=1}^{r}
    [ 
      \mu_k \cdot (\vv_k^\trans \vs)^2
      -
      \E\{\hat \mu_k \cdot (\vhv_k^\trans \vs)^2\}
    ].
\end{align}
We focus our attention on the $k$\textsuperscript{th} summand of the last term.  Write
\begin{align}
  \hat \mu_k  &=  \bar \mu_k  + n^{-1/2} Z_{k}, \label{eqn:muk} \\
  \hat \rho_{kl} &= \bar \rho_{kl} + n^{-1/2} W_{kl},
\end{align}
where
\begin{align}
  \bar \mu_k
    &= (\mu_k + 1)\Big(\frac{m}{n\mu_k} + 1\Big), \label{eqn:mk} \\
  \bar \rho_{kk}^2
    &= \frac{1 - \tfrac{m}{n \mu_k^2}}{1 + \tfrac{m}{n \mu_k}}, \label{eqn:BRresult}
\end{align}
and $\bar \rho_{kl} = 0$ when $k \neq l$.

Theorem 5 of \citet{ona07} gives  that, if $\mu_k > c^{-1/2}$, then
$Z_k$ converges in distribution to a mean-zero normal random variable.
Furthermore, since $\mu_k \neq \mu_l$, Theorem 1 of \citet{ona07} also yields 
that the vector $(W_{kl}, l = 1, \dotsc, r)$ is asymptotically mean-zero multivariate normal
with uncorrelated elements. 
Though not stated explicitly, Onatski's proof shows that $Z_k$ and $W_{kl}$
are asymptotically uncorrelated for $l = 1, \dotsc, r$. 

Next, by Lemma \ref{lem:decomp}, we have $\E(\vtv^\trans \vs) = 0$ and 
\[
\E(\vtv_k^\trans \vs)^2 = \E(\vtv_k^\trans \vs_{\mV^{\perp}})^2
= \frac{1}{m-r} \vs_{\mV^{\perp}}^\trans \vs_{\mV^{\perp}}.
\]
By Lemma \ref{lem:decomp} and Equations \eqref{eqn:muk}~--~\eqref{eqn:BRresult} and Assumption A1, we obtain 
\begin{align}
\E\{ \hat \mu_k \cdot (\vhv_k^\trans \vs)^2 \}
 &= \sum_{l=1}^{r} \{ \bar \mu_k \cdot \bar \rho_{kl}^2 + o(n^{-1/2}) \} \cdot (\vv_l^\trans \vs)^2 \notag \\
 &\qquad+ \{ \bar \mu_k \cdot (1 - \bar \rho_{kk}^2) + o(n^{-1/2}) \} \cdot
\frac{1}{m-r} \vs_{\mV^{\perp}}^\trans \vs_{\mV^{\perp}}.
 \end{align}
 Therefore,
 \begin{align*}
-n \Big(\E( \mu_k & \cdot (\vv_k^\trans \vs)^2) - \E(\hat \mu_k \cdot (\vhv_k^\trans \vs)^2)\Big)
  \\
 & = - n (\mu_k - \bar \mu_k \bar\rho_{kk}^2) (\vv_k^\trans \vs)^2 
  + o(n^{1/2} \cdot \vs_{\mV}^\trans \vs_{\mV}) \\
 &\hspace{2cm} + \bar \mu_k (1 - \bar \rho_{kk}^2) \frac{n}{m}
\vs_{\mV^{\perp}}^\trans \vs_{\mV^{\perp}} + o(n^{-1/2}
\vs_{\mV^{\perp}}^\trans \vs_{\mV^{\perp}}) \\
 &= \vs^\trans \vs \cdot \df_k(\vs) + o(n^{1/2} \vs_{\mV}^\trans \vs_{\mV} +
n^{-1/2} \vs_{\mV^{\perp}}^\trans \vs_{\mV^{\perp}}).
 \end{align*}
 The result now follows by summing over $k$.
\end{proof}
\begin{remark}
 We use Assumption A1 only for making the exposition simple. It can be shown to hold under a broad class of conditions by the arguments of Theorem 5 of \citet{ona07}. Even if Assumption A1 does not hold, Theorem \ref{Thm:signal} holds with a slightly larger error term. Since we use a different estimator instead of the one defined in Equation \ref{E:dfk} (see Section \ref{S:df-estimate}) in our data analysis, we do not pursue this issue further.
 \end{remark}
\begin{corollary}\label{Cor:signal-wrong-r}
Under the hypothesis of Theorem \ref{Thm:signal}, if we instead suppose $ \hat{r} \neq r$, 
then
\begin{equation}
\df(\vs)  = \sum_{k=1}^r \df_k(\vs) + \err(\vs)
          + o\Big(n^{1/2} \sum_{k=1}^r \frac{(\vv_k^\trans \vs)^2}{\vs^\trans \vs} + n^{-1/2} \Big),
\end{equation}
where
\[
  \err(\vs) =
  \begin{cases}
     - n \cdot \sum_{k=\hat r + 1}^{r} \mu_k \cdot \frac{(\vv_k^\trans \vs)^2}{\vs^\trans \vs} &\text{if $\hat r < r$,} \\
     (\hat r - r) \,\Big( 1 + \sqrt{\frac{n}{m}}\Big)^2 &\text{if $\hat r > r$.}
  \end{cases}
\]
\end{corollary}
\begin{proof}
Suppose $\hat{r}< r$. Then  by Equation \eqref{eqn:dfnsdef} in the proof of Theorem \ref{Thm:signal} we obtain
\begin{align} \label{eqn:corhrlr}
(\vs^\trans \vs) \cdot \df(\vs) = 
- n \cdot
  \sum_{k=1}^{\hat{r}}
    \{
      \mu_k \cdot (\vv_k^\trans \vs)^2
      -
      \E(\hat \mu_k \cdot (\vhv_k^\trans \vs)^2)
    \} - n \cdot
  \sum_{k=\hat{r}+1}^{r}
    \{
      \mu_k \cdot (\vv_k^\trans \vs)^2
     \}.
\end{align}
From \eqref{eqn:corhrlr} and Theorem \ref{Thm:signal}, the claim follows for $\hat{r} < r$.\par
Suppose to the contrary that $\hat{r} > r$. A computation similar to the above yields
\begin{align} \label{eqn:corgrlr}
(\vs^\trans \vs) \cdot \err(\vs) = 
n \cdot
  \sum_{k=r+1}^{\hat{r}}
      \E\{\hat \mu_k \cdot (\vhv_k^\trans \vs)^2\}.
\end{align}
Theorem 1 of \citet{ona07}  gives that
\[
\hat \mu_k  = \Big(1 + \sqrt{\frac{m}{n}}\Big)^2 + o_p(1)
\]
and $ \E\{(\vhv_k^\trans \vs)^2\} = \vs^\trans \vs/m$. Summing over $k$ yields the claim and the proof is finished.
\end{proof}

\begin{remark}
The requirement that $\mu_k > c^{-1/2}$ in Theorem ~\ref{Thm:signal} and Corollary~\ref{Cor:signal-wrong-r} is not artificial; there indeed is a phase transition
in the asymptotic behavior of the eigenvalues at $\mu_k = c^{-1/2}$ (see \cite{baik2005phase,
ona07}).
Consequently, Theorem~\ref{Thm:signal} and Corollary~\ref{Cor:signal-wrong-r} do not apply if
some $\mu_k$ is below the phase transition ($\mu_k \leq c^{-1/2}$).  Following
an argument similar to the proof of the $\hat r > r$ case of
Corollary~\ref{Cor:signal-wrong-r}, we conjecture that when $\mu_k \leq
c^{-1/2}$, the degree of freedom term $\df_k(\vs)$ should be defined as
\[
  \df_k(\vs)
  =
  \Big(1 + \sqrt{\frac{n}{m}}\Big)^2
  -
  n
  \mu_k \frac{(\vv^\trans \vs)^2}{\vs^\trans \vs}.
\]
\end{remark}

\section{Simulation Study}\label{S:simulation}

We perform a number of confirmatory simulations to verify the theory in
Section~\ref{S:df}.  In these simulations, we vary the number of rows, $n$,
over the set $\{ 5, 10, 50, 100 \}$ and we vary the number of columns, $m$,
over the set $\{ 5, 10, 50, 100, 500, 1000, 5000, 10000 \}$.  We take the test
direction $\vs$ to be the first standard basis vector $\vs = ( 1, 0, \dotsc,
0)$ in $\reals^m$.  For a given set of simulation parameters, we perform
$10,000$ replicates of the following procedure: \begin{enumerate}

\item Generate data from the model with $r$ latent factors,
\(
  \mY = \sqrt{n} \mU \mD \mV^\trans + \mE,
\)
where the elements of $\mE$ are independent mean-zero normal variates with
variance $\sigma^2 = 1$.  Matrices $\mU$ and $\mV$ have orthonormal columns,
while $\mD$ is diagonal with $(\mD)_{kk}^2 = \mu_k$ for $k = 1, \dotsc, r$.
In each set of simulations, we fix $\mD$ and $\mV$, and we
generate a uniform random $\mU$ for each simulation replicate.

\item Fit the bilinear model with $\hat r$ latent factors via least squares,
\(
  \mhY = \sqrt{n} \mhU \mhD \mhV^\trans.
\)
Compute the residual matrix $\mhE = \mY - \mhY$.

\item Calculate the residual sum of squares along the test direction, 
\(
  \RSS(\vs) = \vs^\trans \mhE^\trans \mhE \vs,
\)
and the observed degrees of freedom along this direction,
\(
  \df(\vs) = n - \RSS(\vs) / \sigma^2.
\)
\end{enumerate}
We estimate $\E\{\df(\vs)\}$ as the average value of $\df(\vs)$ over all
replicates of the simulation; we also compute the standard error of the
estimate via the central limit theorem.  Finally, we compare the theoretical
degrees of freedom estimate to the simulation-based estimate.

\subsection{Noise Case}

In the noise case, we simulate with no true latent factors ($r = 0$), and we
fit with one estimated latent factor using $\hat r = 1$.  The theoretical
degrees of freedom are computed from Theorem~\ref{T:noise-case}.  As can be
seen in Figure~\ref{fig:sim-noise-dfcurve}, the theory
fits well with the simulations when the problem dimensions are large, say for
$n \geq 2500$ (smaller problem dimensions are excluded from the figure).

\begin{figure}
  \centering
  \includegraphics[scale=0.7]{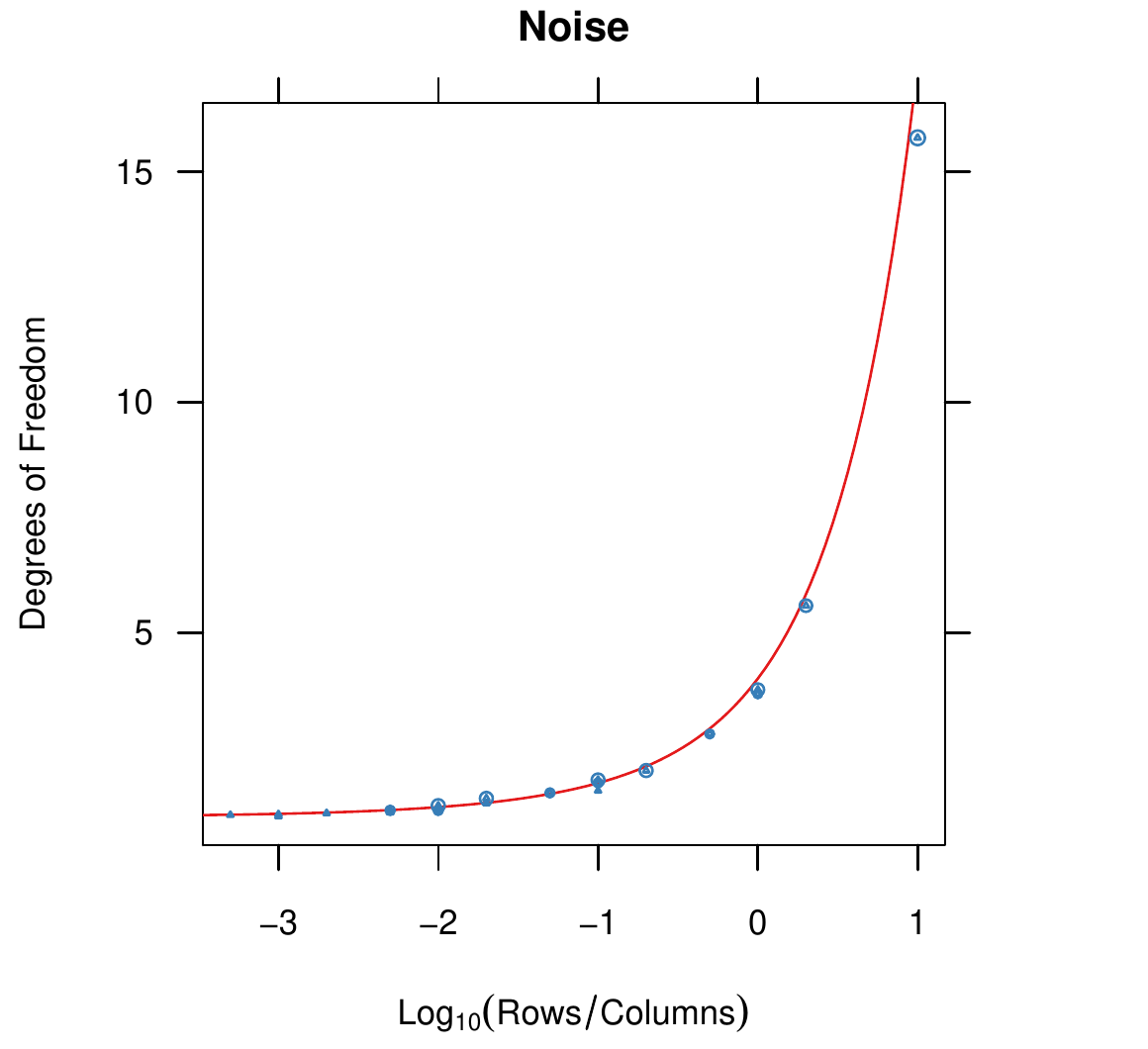}
  \caption{
    Theoretical degrees of freedom for the null case (solid red line) agree
    with the empirical estimates (blue points and circles).  Circle radius shows one
    standard error of the estimates along the $y$-axis.
  }\label{fig:sim-noise-dfcurve}
\end{figure}

\subsection{Signal Case}

For the signal case, the degrees of freedom depend on the signal strength and
true factors.  We simulate $r = 1$ true latent factor with signal strength
$\mu$ varying over the set $\{ 1.0, 1.5, 3.0, 21.0 \}$.  We consider four choices
of the factor loading vector $\vv$:
\begin{description}
  \item[Ones.] $\vv = (1 / \sqrt{m}, \dotsc, 1 / \sqrt{m})$;
  \item[Basis.] $\vv = (1, 0, \dotsc, 0)$;
  \item[Perp.~Ones.] $\vv = (0, 1/\sqrt{m-1}, \dotsc, 1/\sqrt{m-1})$;
  \item[Perp.~Basis.] $\vv = (0, 1, 0, \dotsc, 0)$.
\end{description}
In all cases, $\vv$ is a unit vector.  In the ``Perp.''~cases, $\vv$ is
orthogonal to the test direction $\vs$.

The asymptotic degrees of freedom in each of the four cases are as follows:
\begin{description}
  \item[Ones.]
    \[
      \df(\vs) =
      \begin{cases}
        1 + n/m + \sigma^2 / \mu
          &\text{if $\mu > \sigma^2 \sqrt{m/n}$,} \\
        (1 + \sqrt{n/m})^2 - (\mu / \sigma^2) \cdot \sqrt{n/m}
          &\text{otherwise.}
      \end{cases}
    \]
  \item[Basis.]
    \[
      \frac{\df(\vs)}{n} = 
      \begin{cases}
        1 - m \sigma^2 / (n \mu) - m \sigma^4 / (n \mu^2)
          &\text{if $\mu > \sigma^2 \sqrt{m/n}$,} \\
        - (\mu / \sigma^2)
          &\text{otherwise.}
      \end{cases}
    \]
  \item[Perp.~Ones, Perp.~Basis.]
    \[
      \df(\vs) =
      \begin{cases}
        1 + (\sigma^2 / \mu)^2
          &\text{if $\mu > \sigma^2 \sqrt{m/n}$,} \\
        (1 + \sqrt{n/m})^2
          &\text{otherwise.}
      \end{cases}
    \]
\end{description}
In the ``Basis'' case, we study $\df(\vs) / n$ instead of $\df(\vs)$ so
that the asymptotic limit depends on $n$ only through the ratio $n/m$.
Figure~\ref{fig:sim-signal-dfcurve-all4}
demonstrates that the asymptotic expressions agree with the theory, even for
relatively small sample sizes.

\begin{figure}
  \centering
  \includegraphics[width=\textwidth]{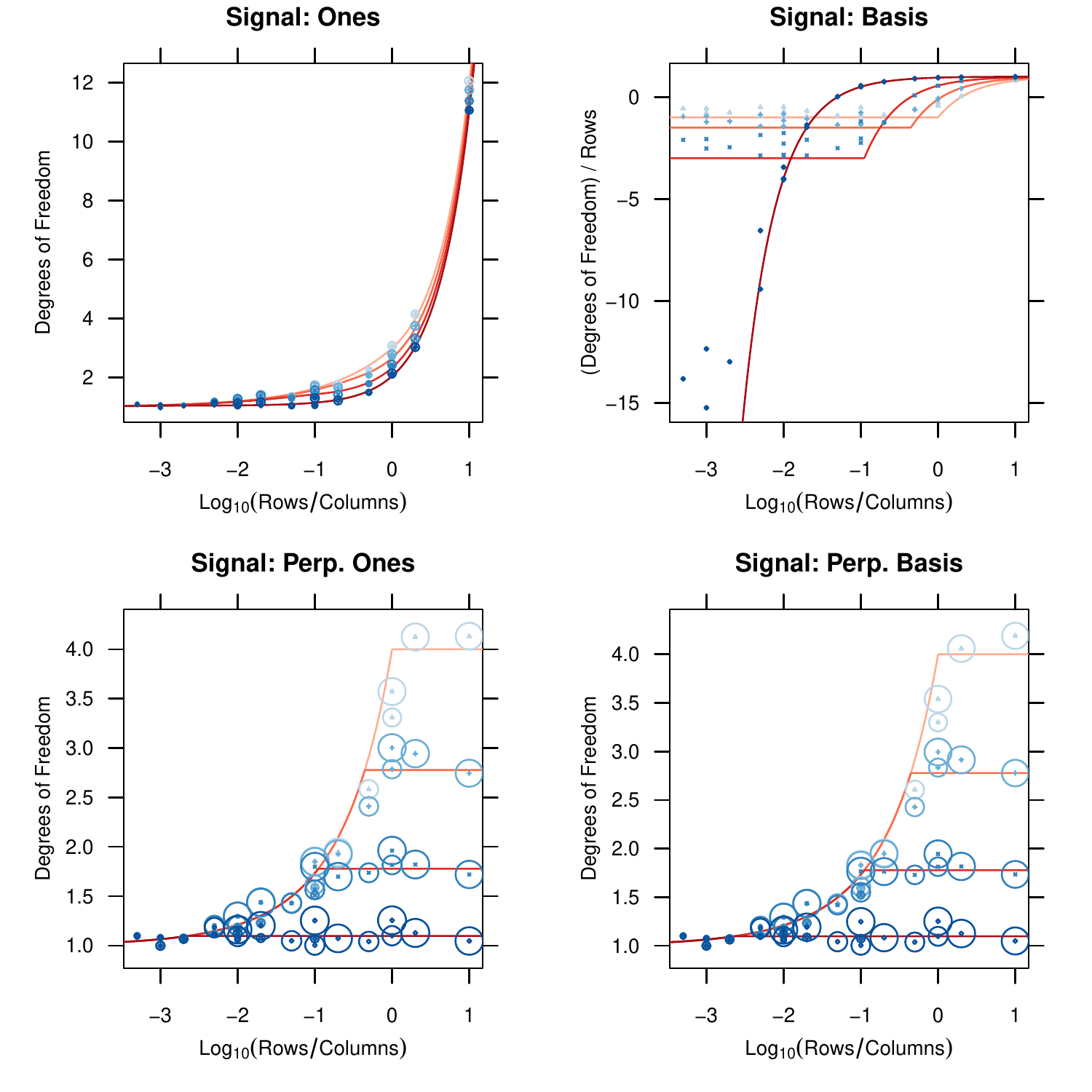}
  \caption{
    Theoretical (solid lines) and empirical estimates (points and circles)
    for the signal case agree (lighter hues correspond to weaker signal strengths).
    When the true signal vector is equal to the test direction (``Signal:
    Basis''), the empirical results fit well with the theory for large sample
    sizes; agreement is better when the signal strength is above the phase
    transition.
    When the true signal vector is orthogonal to the test direction (``Signal:
    Perp.~Ones'' and ``Signal: Perp.~Basis''), degrees of freedom do not depend
    on $m/n$ when signal strength is above the phase transition.
%
  }
  \label{fig:sim-signal-dfcurve-all4}
\end{figure}

\subsection{Agreement with Chi-Squared Distribution}

For each of the simulation settings considered, we compare the distribution of
the degrees-of-freedom adjusted residual sum of squares with the corresponding
$\chi^2$ distribution.  For example, in the noise simulation with $m = 1000$
and $n = 10$, we ran $10000$ replicates of the simulation.  For each
simulation, we computed a residual sum of squares value.  We then compared the
empirical quantiles of the $10000$ values with the quantiles of a $\chi^2$
distribution with degrees of freedom predicted by Theorem~\ref{T:noise-case}.
Figure~\ref{fig:sim-noise-qq} shoes a quantile-quantile plot of the results.
We can see good agreement between the empirical and the theoretical
distribution. Note that, when the residual sum of squares have the posited chi-squared distribution, our test statistic has a $t$ distribution.

\begin{figure}
  \centering
  \includegraphics[scale=0.7]{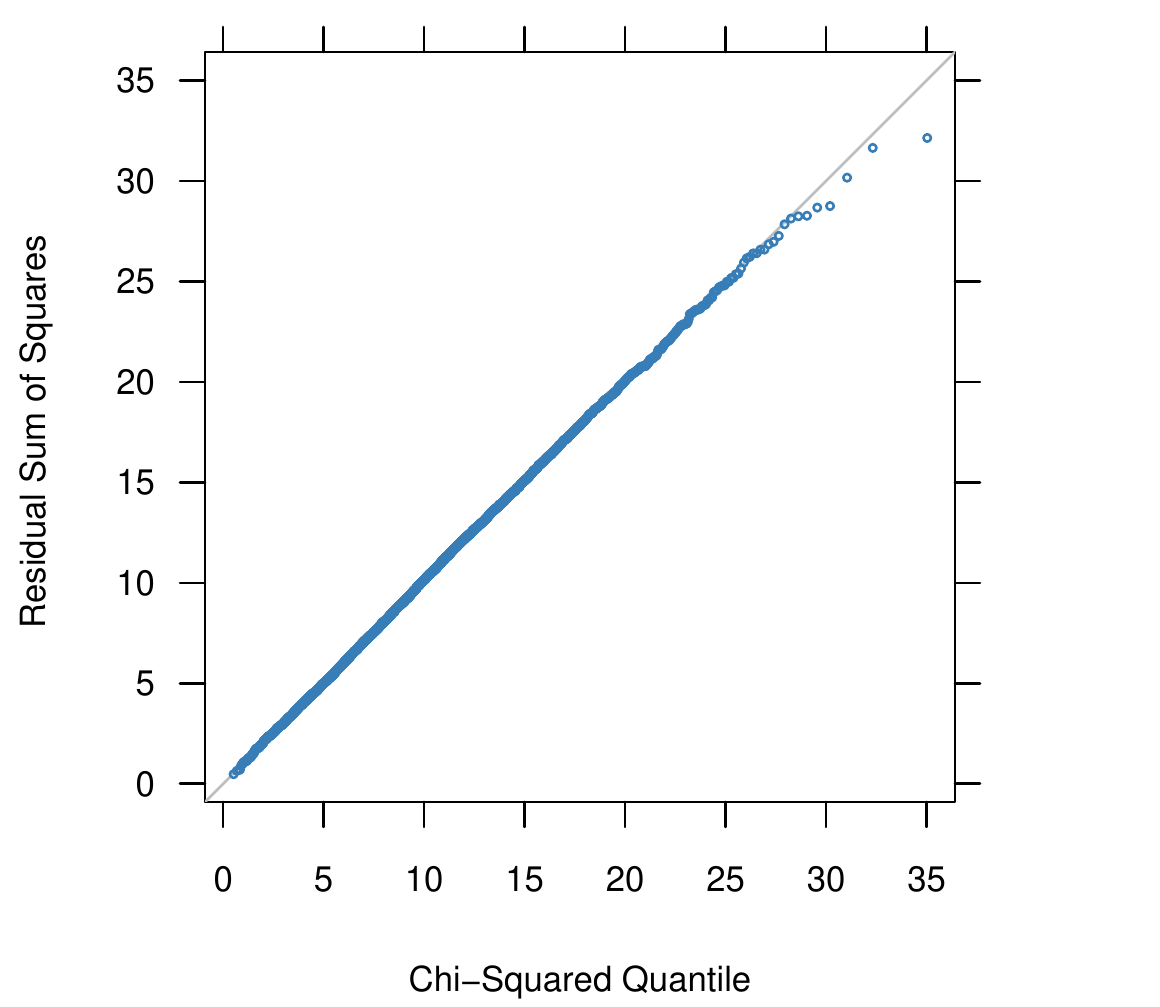}
  \caption{
    There is good agreement between the chi squared distribution with the
    degrees of freedom predicted by Theorem~\ref{T:noise-case}.
  }
  \label{fig:sim-noise-qq}
\end{figure}

We each value of $m$ and $n$, we computed a Kolmogorov-Smirnoff statistic for
testing the hypothesis that the residual sum of squares follows a $\chi^2$
distribution with degrees of freedom predicted by Theorem~\ref{T:noise-case}.  The
$p$-values from the test are large whenever $n$ is above $50$ and $m$ is above
$500$.  Table~\ref{tab:noise-ks-pval} shows the results.

\begin{table}
\caption{Kolmogorov-Smirnoff $p$-value for $\chi^2$ goodness of fit (Noise)}
\label{tab:noise-ks-pval}
\begin{tabular}{ccccccccc}
\toprule
& \multicolumn{8}{c}{Columns ($m$)} \\
\cmidrule(l){2-9}
Rows ($n$) & 5 & 10 & 50 & 100 & 500 & 1000 & 5000 & 10000  \\
\midrule
5 & 0.00 & 0.00 & 0.00 & 0.00 & 0.00 & 0.00 & 0.53 & 0.39  \\
10 & 0.00 & 0.00 & 0.00 & 0.00 & 0.07 & 0.09 & 0.24 & 0.49  \\
50 & 0.00 & 0.00 & 0.00 & 0.01 & 0.89 & 0.78 & 0.30 & 0.28  \\
100 & 0.00 & 0.00 & 0.01 & 0.15 & 0.84 & 0.57 & 0.49 & 0.63  \\
\bottomrule
\end{tabular}

\end{table}

We computed analogous Kolmogorov-Smirnoff goodness of fit $p$-values for the
signal simulations.  Specifically, for each value of $m$, $n$, $\mu$, and for
each choice of the signal direction, we compared the distribution of the
residual sum of squares from the $10000$ replicates with the $\chi^2$
distribution having degrees of freedom predicted by
Theorem~\ref{Thm:signal}.  Appendix~A (\ref{supp}) contains the result tables
analogous to Table~\ref{tab:noise-ks-pval}.  As with the noise case, in most signal settings, regardless of the
signal strength, we see large $p$-values whenever $n$ is above $50$ and $m$ is
above $500$.  However, when the test direction is parallel to the signal
direction (the ``Signal: Basis'' case) the $p$-value from the
Kolmogorov-Smirnoff test is always below $0.01$.  This suggests that the
$\chi^2$ distribution is a poor fit when the test direction is parallel or
highly correlated with the signal direction.  In other situations with
moderate $n$ and $m$, our simulations show good agreement with the $\chi^2$
distribution.

\section{Estimating Degrees of Freedom in Applications: A Conservative Estimator}\label{S:df-estimate}

The main result of Section~\ref{S:df} is that the asymptotic
degrees of freedom associated with the $k$\textsuperscript{th} latent factor is
given by
\begin{equation}\label{eqn:dfk-full}
  \df_k(\vs)
  =
  \begin{cases}
    n
    \Big( 1 - \frac{m \sigma^2}{n \mu_k} - \frac{m \sigma^4}{n \mu_k^2} \Big)
    \frac{(\vv_k^\trans \vs)^2}{\vs^\trans \vs}
    +
    \Big( 1 + \frac{\sigma^2}{\mu_k} \Big)^2
    \Big( 1 - \frac{(\vv_k^\trans \vs)^2}{\vs^\trans \vs} \Big)
      &\text{if $\mu_k > \sigma^2 \sqrt{m/n}$,} \\
    (1 + \sqrt{n/m})^2 - n \frac{\mu_k}{\sigma^2} \frac{(\vv_k^\trans \vs)^2}{\vs^\trans \vs}
      &\text{otherwise.}
  \end{cases}
\end{equation}
This result, while theoretically interesting, is not directly applicable to
data analysis.  For practical purposes, we need an estimate of $\df(\vs)$
which does not depend on unknown quantities.  

A plug-in estimator (replacing population quantities $\mu_k$, $\sigma^2$, and
$\vv_k^\trans \vs$ with the corresponding sample-based quantities) is likely
to under-estimate $\df_k(\vs)$ since, almost surely, $\hat \mu_k > \mu_k$ and
$(\vhv_k^\trans \vv_k)^2 < 1$.  Under-estimating $\df_k(\vs)$ leads to smaller
estimates of $\sigma^2(\vs)$, which in turn leads to higher $t$-statistics and
more false discoveries.

We propose a conservative estimator for $\df(\vs)$.  
First, from~\eqref{eqn:dfk-full}, we have the upper bound
\[
  \df_k(\vs)
    \leq
      n \Big( 1 - \frac{m \sigma^2}{n \mu_k} - \frac{m \sigma^4}{n \mu_k^2} \Big)
      \frac{(\vv_k^\trans \vs)^2}{\vs^\trans \vs}
      +
      (1 + \sqrt{n/m})^2.
\]
Next, we note that
\[
  (\vhv_k^\trans \vs)^2
    \geq \bar \rho_{kk}^2 (\vv_k^\trans \vs)^2 + O_P(n^{-1/2})
    \geq
      \Big( 1 - \frac{m \sigma^4}{n \mu_k^2} \Big) (\vv_k^\trans \vs)^2
      + O_P(n^{-1/2}).
\]
Therefore, the estimator
\begin{equation}\label{eqn:dfhat}
  \widehat{\df}_k(\vs)
  =
  n \frac{(\vhv_k^\trans \vs)^2}{\vs^\trans \vs}
  +
  (1 + \sqrt{n/m})^2
\end{equation}
is asymptotically greater than $\df_k(\vs)$.

Even though the estimator is conservative, the difference
\(
  \widehat{\df}_k(\vs) - \df_k(\vs)
\)
is small in regimes of practical interest, when $\mu_k$ is well
above the phase transition, \emph{i.e.}, $\mu_k \gg \sigma^2 \sqrt{m/n}$).

\section{Degrees of Freedom correction to the AGEMAP study}\label{S:AGEMAP}

The analysis of the AGEMAP dataset conducted by \citet{Zahn:07} was criticized
for its low power in light of the fact that it did not find statistically
significant evidence of many age-related genes in the cortical tissues,
despite ``extensive other evidence on the susceptibility to aging'' of those
tissues \citep{Land:08}.  Indeed, without any adjustment for latent factors,
we find only 19 out of 17864 genes to be significantly age-related at level
$0.001$, roughly the same number to be expected by chance under the null
hypothesis of no age-related genes.  In this section, we perform an
analysis that adjusts for latent factors, and we show that this leads to many
more significant findings.

We fit the bilinear model to the AGEMAP
dataset described in Section~\ref{S:introduction}.  For each gene, our goal is
to assess the relationship between log activation and age after adjusting for
observed and latent subject-specific covariates.  For gene~$j$, we take
test direction $\vs_j = (\mI - \mH_{\mZ}) \ve_j$, where $\ve_j \in \reals^m$
denotes the $j$\textsuperscript{th} basis vector.  Using a bilinear model to adjust for
observed and latent subject-specific covariates, we perform a test on
$[\mB^\trans \vs_j]_3$, the identifiable component of the age coefficient for
gene $j$.

We first regress gene response on the observed covariates (subject age and
sex; gene tissue type).  An investigation of the residuals from this bilinear
multiple regression fit reveals that two latent factors explain 51.3\% of
the residual variance (Table~\ref{T:agemap-scree}).  After adding these two
estimated latent factors to the regression model, there is no obvious
low-dimensional structure in the residuals.

\begin{table}[h]
\centering

\caption{\textsc{First Two Factors Explains Most of Residual Variance.} Residual
variance explained by each principal component.  A large proportion of the
total variance is explained by the first two components.}
\label{T:agemap-scree}

\vspace{\baselineskip}

\footnotesize
\hfill
\begin{tabular}{lrr}
\toprule
& & \multicolumn{1}{l}{Resid.}\\
\multicolumn{1}{l}{Factor} & \multicolumn{1}{c}{Var.~\%} & \multicolumn{1}{l}{Var.~\%}\\
\midrule
1 & 37.1 \phantom{x}& 62.9 \phantom{}\\
2 & 14.2 \phantom{x}& 48.7 \phantom{}\\
3 &  5.8 \phantom{x}& 42.9 \phantom{}\\
4 &  4.3 \phantom{x}& 38.7 \phantom{}\\
5 &  3.7 \phantom{x}& 34.9 \phantom{}\\
6 &  3.4 \phantom{x}& 31.5 \phantom{}\\
7 &  2.8 \phantom{x}& 28.6 \phantom{}\\
8 &  2.2 \phantom{x}& 26.5 \phantom{}\\
9 &  2.0 \phantom{x}& 24.4 \phantom{}\\
10 &  1.9 \phantom{x}& 22.5 \phantom{}\\
11 &  1.8 \phantom{x}& 20.7 \phantom{}\\
12 &  1.5 \phantom{x}& 19.2 \phantom{}\\
\bottomrule
\end{tabular}
\hfill\hfill
\begin{tabular}{lrr}
\toprule
& & \multicolumn{1}{l}{Resid.}\\
\multicolumn{1}{l}{Factor} & \multicolumn{1}{c}{Var.~\%} & \multicolumn{1}{l}{Var.~\%}\\
\midrule
13 &  1.4 \phantom{x}& 17.8 \phantom{}\\
14 &  1.2 \phantom{x}& 16.6 \phantom{}\\
15 &  1.2 \phantom{x}& 15.5 \phantom{}\\
16 &  1.1 \phantom{x}& 14.3 \phantom{}\\
17 &  1.0 \phantom{x}& 13.3 \phantom{}\\
18 &  1.0 \phantom{x}& 12.3 \phantom{}\\
19 &  0.9 \phantom{x}& 11.4 \phantom{}\\
20 &  0.9 \phantom{x}& 10.5 \phantom{}\\
21 &  0.8 \phantom{x}&  9.6 \phantom{}\\
22 &  0.8 \phantom{x}&  8.8 \phantom{}\\
23 &  0.8 \phantom{x}&  8.0 \phantom{}\\
24 &  0.8 \phantom{x}&  7.3 \phantom{}\\
\bottomrule
\end{tabular}
\hfill\hfill
\begin{tabular}{lrr}
\toprule
& & \multicolumn{1}{l}{Resid.}\\
\multicolumn{1}{l}{Factor} & \multicolumn{1}{c}{Var.~\%} & \multicolumn{1}{l}{Var.~\%}\\
\midrule
25 &  0.7 \phantom{x}&  6.5 \phantom{}\\
26 &  0.7 \phantom{x}&  5.8 \phantom{}\\
27 &  0.7 \phantom{x}&  5.1 \phantom{}\\
28 &  0.7 \phantom{x}&  4.5 \phantom{}\\
29 &  0.6 \phantom{x}&  3.8 \phantom{}\\
30 &  0.6 \phantom{x}&  3.2 \phantom{}\\
31 &  0.6 \phantom{x}&  2.6 \phantom{}\\
32 &  0.6 \phantom{x}&  2.1 \phantom{}\\
33 &  0.6 \phantom{x}&  1.5 \phantom{}\\
34 &  0.5 \phantom{x}&  1.0 \phantom{}\\
35 &  0.5 \phantom{x}&  0.5 \phantom{}\\
36 &  0.5 \phantom{x}&  0.0 \phantom{}\\
\bottomrule
\end{tabular}
\hfill

\end{table}

We obtain conservative degree of freedom estimate $\widehat{\df}_k(\vs_j)$ for the
$\hat r = 2$ estimated latent factors using the estimator~\eqref{eqn:dfhat}.
Figure~\ref{F:agemap-df-comparison} shows these estimates; we can see that
Gollob's method and Mandel's method are both more liberal than our proposed
method.  Our method assigns between $1.1$ and $1.6$ degrees of freedom for
each latent factor, depending on the gene.  The genes with higher assigned
degrees of freedom are the ones with higher loadings for the estimating latent
factors.

\begin{figure}
\centering
\includegraphics[scale=0.7]{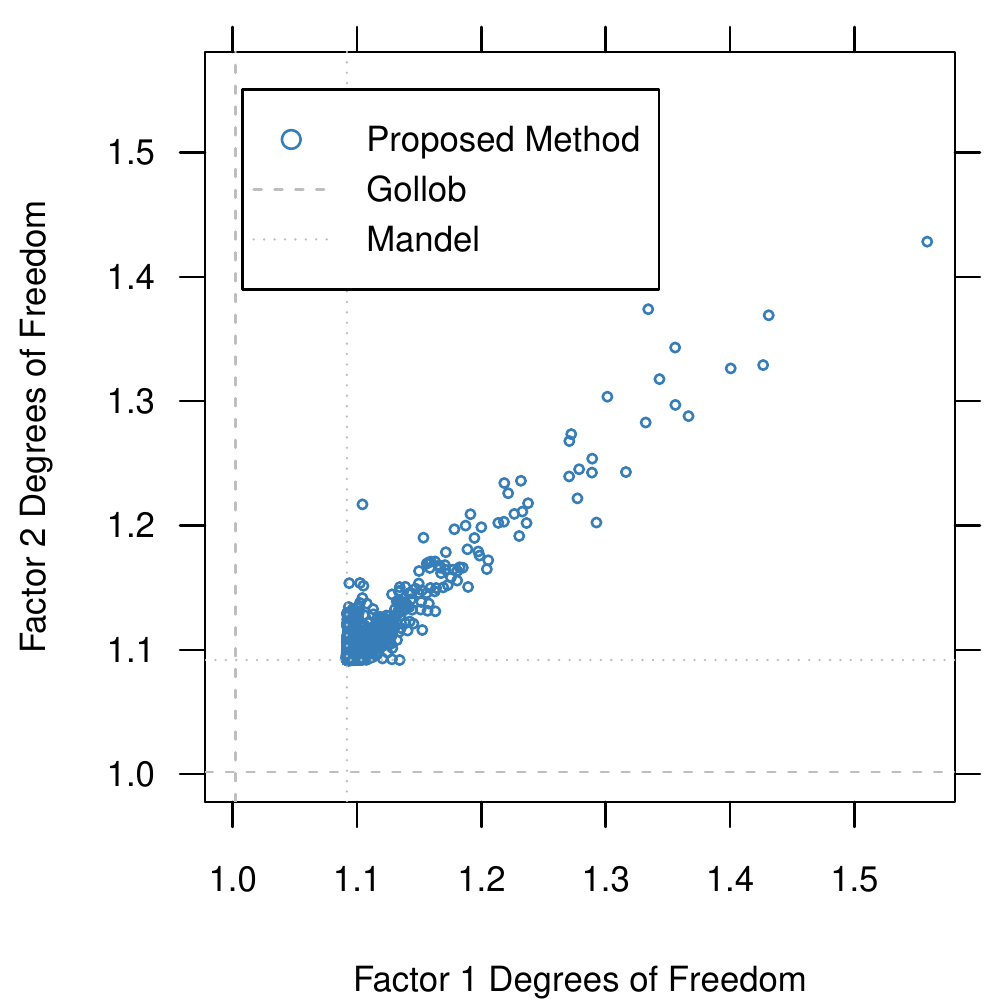}
\caption{\textsc{Degrees of Freedom Method Comparison.}
Gene-specific degrees of freedom estimates for the first two estimated
factors.  Gollob's method and Mandel's method assign the same degrees of
freedom to each gene.  Our proposed method assigns more degrees of freedom to
genes with high factor loadings.
}
\label{F:agemap-df-comparison}
\end{figure}

With the degrees of freedom estimates, we derive a gene-specific error variance
estimate
\[
  \hat \sigma^2 (\vs_j) = \RSS(\vs_j) / \{ n - \widehat{\df}(\vs_j) \},
\]
with
\(
  \widehat{\df}(\vs_j) = \sum_{k=1}^{\hat r} \widehat{\df}_k(\vs_j)
\)
and $n = N - p$.  This, in turn, can be used to compute a test statistic
for $[\mB^\trans \vs_j]_3$.

\begin{figure}
\centering
\includegraphics[scale=0.7]{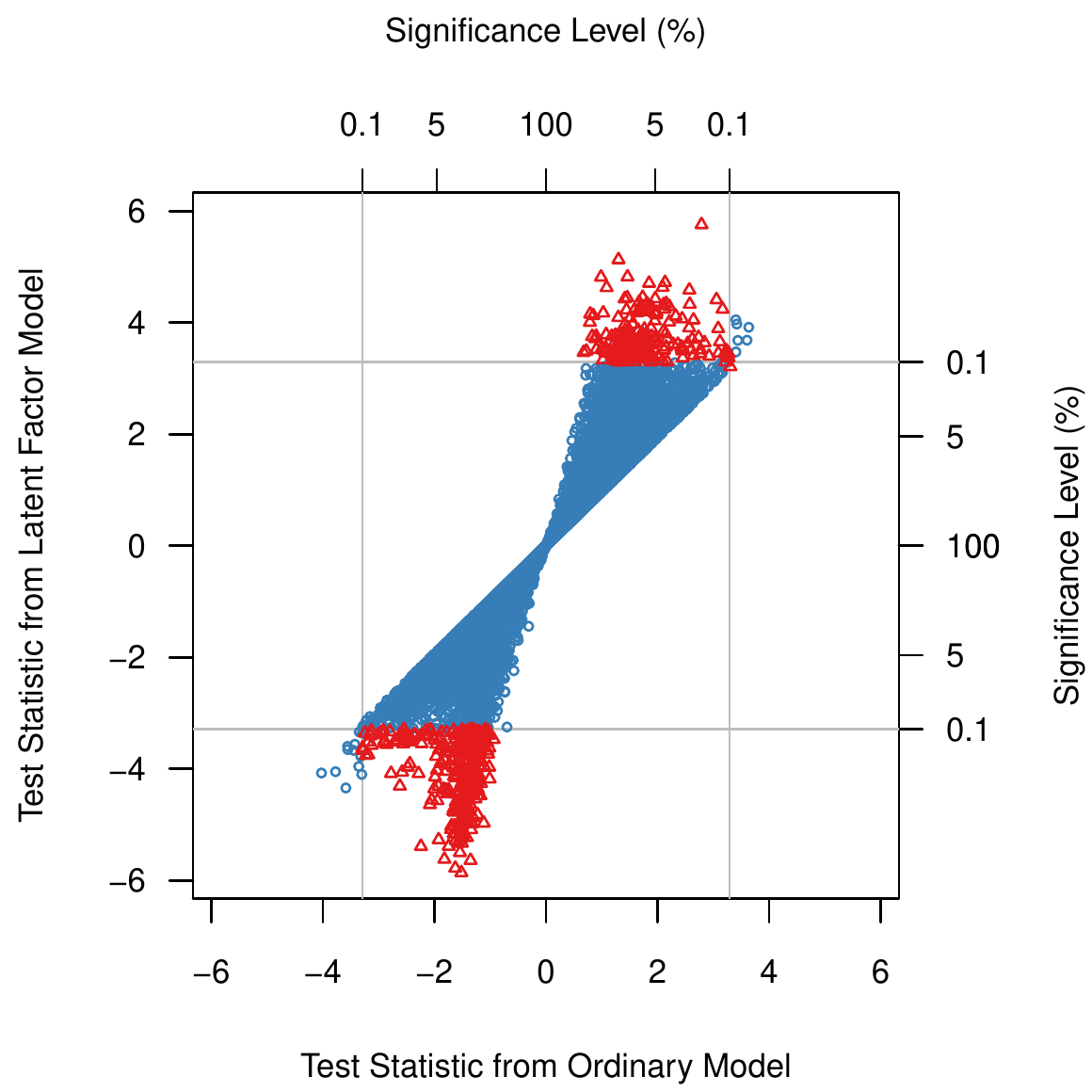}
\caption{\textsc{Latent Factor Model Leads to Different Conclusions.}
Gene-specific regression coefficient $t$ statistics for Age under the ordinary
regression model and the latent factor model with $\hat r = 2$ estimated factors.
There are 496 coefficients which are significant at level~$0.001$ in the
latent factor model but not in the ordinary regression model; there is 1 coefficient
significant at this level in the ordinary regression model but not the latent
factor model.
}
\label{F:agemap-tstat}
\end{figure}

After adjusting for latent factors, there are 514 age coefficients out of
17,864 which are significant at level 0.001.  Without the latent factor
adjustment, we would find only 19 genes to be significant at that level.
Figure~\ref{F:agemap-tstat} shows the test statistics from the model with no
estimated factors ($\hat r = 0$) and the model with ($\hat r = 2$).  For most
genes (85\%), adjusting for latent factors results in a larger test
statistic.  According to our conservative estimate, adjusting for 2 latent
factors uses between $2.184$ and $2.987$ degrees of freedom, depending on the gene.
Contrast this with Gollob's parameter counting scheme, which would assign
$2.004$ degrees of freedom, and Mandel's scheme, which would assign $2.184$
degrees of freedom to all genes.

\subsection{Estimating the False Discovery Rate}

To estimate the false discovery rate, we performed a parametric bootstrap
simulation:

\begin{enumerate}
\item First, we fit a bilinear regression model to the AGEMAP data
using $K = 2$ latent factors.  

\item For the $17350$ genes whose estimated age coefficients
were not significant at level $0.001$, we set the estimates to 0; this left
$514$ nonzero age coefficients.  

\item We simulated $1000$ bootstrap datasets using
the estimated coefficients and latent factors, with a diagonal covariance
matrix for gene-specific regression errors with variances estimated from the
data.  

\item For each bootstrap dataset, we re-fit the model.  We computed the number
of declared significant genes at nominal level $0.001$, using all four degree
of freedom correction methods.  We also fit a model without estimating any
latent factors.

\item We average the false discovery rate (FDR), level/false positive rate
(FDR), and power/true positive rate (TPR).

\end{enumerate}
Table~\ref{tab:fdr-table} summarizes the results.

\begin{table}\label{tab:fdr-table}
\begin{tabular}{lrrrr}
Correction & FDR (\%) & Level/FPR (\%) & Power/TPR (\%) \\
\hline
Proposed Method
       &   4.25 &  0.10 &  74.27 \\
Gollob &   4.27 &  0.10 &  74.32 \\
Mandel &   4.25 &  0.10 &  74.28 \\
Naive  &   4.27 &  0.10 &  74.32 \\
None   &  16.25 &  0.03 &   5.55
\end{tabular}
\caption{Estimated False Discovery Rate (FDR), False Positive Rate (FPR), and
True Positive Rate (TPR) from AGEMAP bootstrap simulation.
The FDR for the ``None''
method had a standard error of $0.2\%$; all other standard errors were below
$0.06\%$.}
\end{table}

We can see that not adjusting for the latent factors results in a lower power
and a higher false discovery rate. All of the other methods give
similar results.

\section{Discussion}\label{S:summary}
Motivated by the AGEMAP study, we have shown how to adjust for latent sources of variability in multivariate regression problems by proposing a simple degrees of freedom assignment for
estimated latent factors.  Our methodology gives a principled alternative to
\emph{ad-hoc} approaches in common use.  We have thus bridged the gap between
theory and practice in this context by proposing a conservative estimate for
the degrees of freedom.
Although our estimator is conservative, it is close to the exact theoretical
value in regimes of common interest, with many responses and strong latent
signals.  Moreover, it is quite simple to apply, and thus ideal for routine
use.

In order to gain theoretical insights, we have made two main
simplifying assumptions.  First, we have assumed that the regression errors
are normally-distributed.  Second, we have assumed that the noise covariance
is a multiple of the identity.  In light of many universality results in
random matrix theory \citep{pillai2011, benaych2011eigenvalues}, the first
assumption (normality) can likely be weakened.  The second assumption is
harder to tackle analytically, but we believe our results hold as long as the
eigenvalues of the error covariance matrix are small relative to the latent
signal strength.  A rigorous analysis of the extent to which this assumption
can be weakened is an area for further research.

\appendix

\section*{Appendix}
\begin{proof}[Proof of Lemma \ref{lem:RSS}]
By construction,
\[
  \mhE^\trans \mhE
  = (\mY - \mhY)^\trans(\mY - \mhY).
  \]
Since the factors were
estimated from the singular value decomposition of $\mY$, they are orthogonal
to the residual matrix.  That is, $\mhE^\trans \mhU = \mzero$ and
$\mhE \mhV^\trans = \mzero$ and hence $\mhE^\trans \mhY = 0$.  Thus,
\begin{align*}
  \mhE^\trans \mhE
  &=  \mY \mY^\trans - 2\mY^\trans \mhY + \mhY^\trans \mhY\\
  &=  \mY \mY^\trans - 2(\mhY + \mhE)^\trans \mhY + \mhY^\trans \mhY \\
  &=  \mY \mY^\trans - \mhY^\trans \mhY \\
  &=   \mY \mY^\trans - n \mhV \mhD^2 \mhV^\trans
\end{align*}
and
\begin{align}
\RSS(\vs)
    &\equiv \vs^\trans  \mY \mY^\trans \vs - n \,\vs^\trans \mhV \mhD^2 \mhV^\trans \vs.
 \end{align}
 Now the result follows from expanding the terms, and using the identity
 \[
 \vs^\trans \mV \mD^2 \mV \vs = \sum_{k=1}^{r}
      \mu_k \cdot (\vv_k^\trans \vs)^2
 \]
 along with an analogous expansion for $\vs^\trans \mhV \mhD^2 \mhV^\trans \vs$.
\end{proof}

\begin{proof}[Proof of Lemma~\ref{lem:decomp}]
Suppose that $\mY = \sqrt{n} \, \mU \mD \mV^\trans + \mE$, where the rows of
$\mE$ are independent mean-zero multivariate normal random vectors with
covariance matrix $\mSigma = \sigma^2 \mI$.  Let $\mY = \sqrt{n} \, \mhU  \mhD
\mhV^\trans$ be a (scaled) singular value decomposition of $\mY$.  Set
$\matrixsymbol{\mathcal{V}}_1 = \mV$ and choose $\matrixsymbol{\mathcal{V}}_2$
such that
\(
  \matrixsymbol{\mathcal{V}}
  =
  [
  \begin{matrix}
    \matrixsymbol{\mathcal{V}}_1 &
    \matrixsymbol{\mathcal{V}}_2
  \end{matrix}
  ]
\)
is an orthogonal matrix.  The matrix $\mhV$ can be decomposed as
\(
  \mhV
  =
  \matrixsymbol{\mathcal{V}}_1 \mhV_1
  +
  \matrixsymbol{\mathcal{V}}_2 \mhV_2,
\)
where $\mhV_l = \matrixsymbol{\mathcal{V}}_l^\trans \mhV$ for $l = 1, 2$.
The claim will follow if we show that the distribution of $\mhV_2$ is
invariant under multiplication on the left by any orthogonal matrix, \emph{i.e.},
if $\mO \mhV_2 \eqd \mhV_2$ for every orthogonal
$\mO$.

 Set $\mE_l = \mE \matrixsymbol{\mathcal{V}}_l^\trans$ for $l = 1, 2$.
Note that $\mE_2 \mO^\trans \eqd \mE_2$.
Write
\(
  \mY
  = 
  (\sqrt{n} \, \mU \mD + \mE_1)
  \matrixsymbol{\mathcal{V}}_1^\trans
  +
  \mE_2
  \matrixsymbol{\mathcal{V}}_2^\trans.
\)
Set
\(
  \mY'
  =
  (\sqrt{n} \, \mU \mD + \mE_1)
  \matrixsymbol{\mathcal{V}}_1^\trans
  +
  \mE_2 \mO^\trans
  \matrixsymbol{\mathcal{V}}_2^\trans
\)
and let
\(
  \mY'
  =
  \sqrt{n} \,
  \mhU'
  \mhD'
  \mhV'^\trans
\)
be the singular value decomposition of $\mY'$.  Since $\mY' \eqd \mY$,
it must follow that
\(
  \matrixsymbol{\mathcal{V}}_2^\trans \mhV'
  \eqd
  \mhV_2.
\)
In fact,
\(
  \mhV'
  =
  \matrixsymbol{\mathcal{V}}_1 \mhV_1
  +
  \matrixsymbol{\mathcal{V}}_2 \mO \mhV_2
\)
by construction since
\(
  \mY' 
  (
  \matrixsymbol{\mathcal{V}}_1 \mhV_1
  +
  \matrixsymbol{\mathcal{V}}_2 \mO \mhV_2
  )
  =
  \mY
  \mhV.
\)
Therefore,
\(
  \matrixsymbol{\mathcal{V}}_2^\trans \mhV = \mO \mhV_2
\)
and thus
\(
  \mO \mhV_2 \eqd \mhV_2,
\)
finishing the proof.
\end{proof}

\section*{Acknowledgements}
The authors thank Art Owen for suggesting the research problem and Paul
Bourgade for helpful discussions.  Natesh Pillai was partially supported by 
 National Science Foundation under the grant DMS~1107070.

\bibliographystyle{imsart-nameyear}

\bibliography{refs}

\includepdf[pages={1-5}]{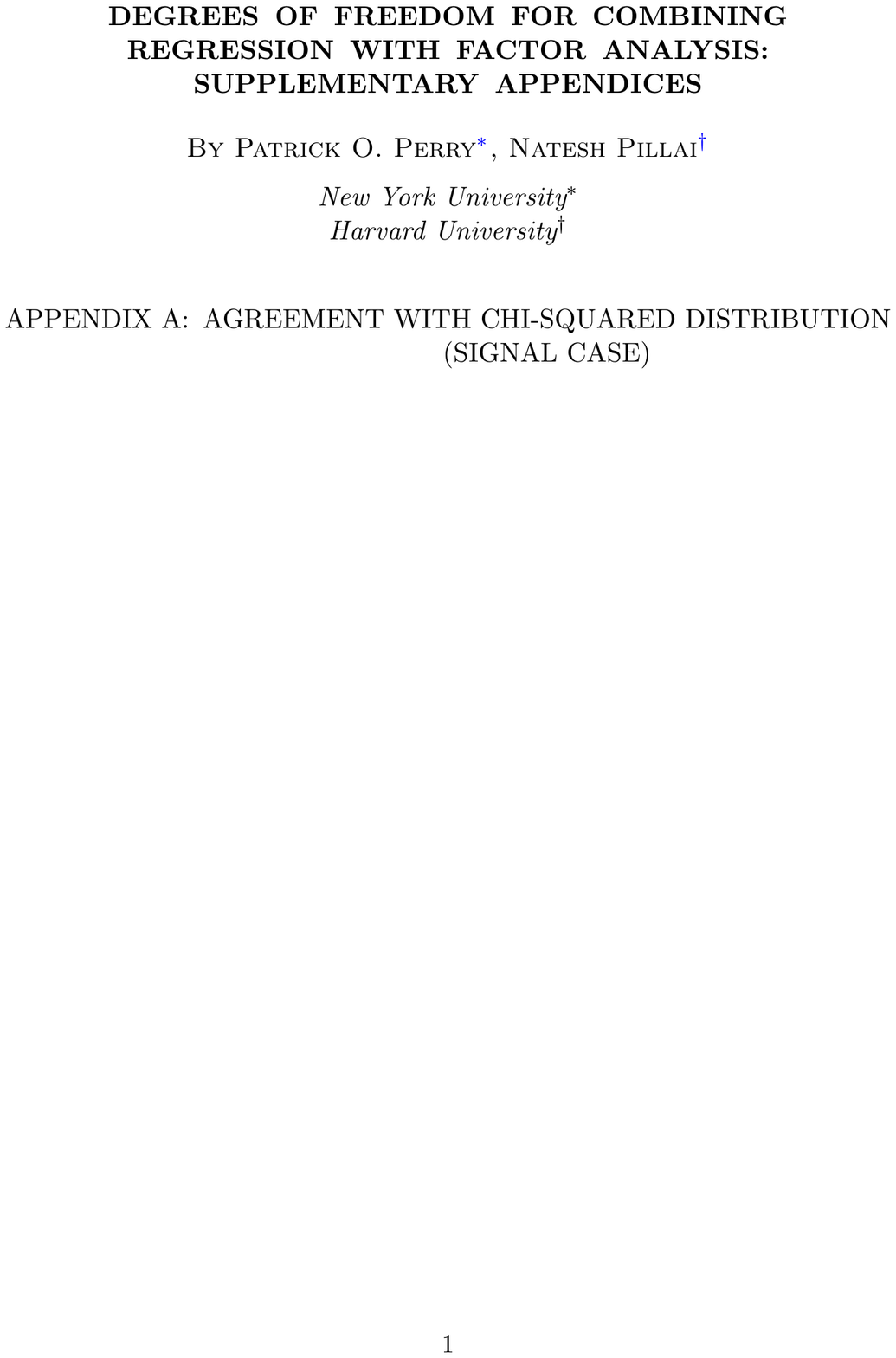} 



\end{document}